\documentclass[sigconf]{acmart}

\usepackage{booktabs} 
\usepackage{tabularx}
\usepackage{lipsum,graphicx,subcaption}
\usepackage{algorithm} 
\usepackage{algorithmic} 
\usepackage{enumitem}

\usepackage[notheorems]{dgleich-math}

\usepackage{ushort}
\newcommand{\cel}[1]{\ushort{#1}}

\newcommand{\celm}[1]{\cel{\mat{#1}}}

\newcommand{\cP}{\ensuremath{\cel{P}}}

\providecommand{\cmP}{\ensuremath{\celm{P}}}

\newcolumntype{Y}{>{\raggedright\arraybackslash}X}
\newcolumntype{W}{>{\raggedleft\arraybackslash}X}
\newcolumntype{Z}{>{\centering\arraybackslash}X}
\newcolumntype{H}{>{\setbox0=\hbox\bgroup}c<{\egroup}@{}}

\setcopyright{none}

\acmDOI{}

\acmISBN{}

\acmYear{}
 
\acmPrice{}

\graphicspath{{figure/}}

\begin{document}
	\title{Retrospective Higher-Order Markov Processes for User Trails}

\author{Tao Wu}
\affiliation{%
  \institution{Purdue University}
  \city{West Lafayette} 
  \state{IN} 
}
\email{wu577@purdue.edu}


\author{David F. Gleich}
\affiliation{%
  \institution{Purdue University}
  \city{West Lafayette} 
  \state{IN} 
}
\email{dgleich@purdue.edu}



\begin{abstract}
Users form information trails as they browse the web, checkin with a geolocation, rate items, or consume media. A common problem is to predict what a user might do next for the purposes of guidance, recommendation, or prefetching. First-order and higher-order Markov chains have been widely used methods to study such sequences of data. First-order Markov chains are easy to estimate, but lack accuracy when history matters. Higher-order Markov chains, in contrast, have too many parameters and suffer from overfitting the training data. Fitting these parameters with regularization and smoothing only offers mild improvements.  In this paper we propose the retrospective higher-order Markov process (RHOMP) as a low-parameter model for such sequences. This model is a special case of a higher-order Markov chain where the transitions depend retrospectively on a single history state instead of an arbitrary combination of history states. There are two immediate computational advantages: the number of parameters is linear in the order of the Markov chain and the model can be fit to large state spaces. Furthermore, by providing a specific structure to the higher-order chain, RHOMPs improve the model accuracy by efficiently utilizing history states without risks of overfitting the data. We demonstrate how to estimate a RHOMP from data and we demonstrate the effectiveness of our method on various real application datasets spanning geolocation data, review sequences, and business locations. The RHOMP model uniformly outperforms higher-order Markov chains, Kneser-Ney regularization, and tensor factorizations in terms of prediction accuracy.

\end{abstract}

%
%



\keywords{Higher-order Markov chains; Tensor factorization; User models}

\maketitle


\section{Introduction} 
\label{sec_intro}

User trails record sequences of activities when individuals interact with the Internet and the world. Such data come from various applications when users write a product review~\cite{mcauley2013amateurs}, checkin at a physical location~\cite{cho2011friendship,yang2013fine}, visit a webpage, or listen to a song~\cite{celma2009music}. Understanding the properties and predictability of these data helps improve many downstream applications including overall user experiences, recommendations, and advertising~\cite{figueiredo2016tribeflow,awad2012prediction}. We study the prediction problem and our goal is to estimate a model to describe and predict a set of user trails. 

Markov chains are one of the most commonly studied models for this type of data. For these models, each checkin place, website, or song is a state. Users transition among these states following Markov rules. In a first-order Markov model, the transition behavior to the next state of the sequence only depends on the current state. Higher-order Markov models include a more-realistic dependence on a larger number of previous states, and multiple recent studies found that first-order Markov chains do not fully capture the user behaviors in web browsing, transportation and communication networks~\cite{rosvall2014memory,chierichetti2012web}. Furthermore ignoring the effects of second-order Markov dynamics has significant negative consequences for downstream applications including community detection, ranking, and information spreading~\cite{rosvall2014memory,Benson-2016-motif-spectral}.

The downside to higher-order Markov models is that the number of parameters grows exponentially with the order. (If there are $N$ states and we model $m$ steps of history, there are $N^{m+1}$ parameters.) So, even if we could accurately learn the parameters, it is already challenging to even store them. (Some practical techniques include low-rank and sparse approximations, but these pose their own problems.) Second, since the number of model parameters grows rapidly, the amount of training data required also grows exponentially with the order $m$~\cite{chierichetti2012web}. Acquiring such huge amounts of training data is usually impossible. Lastly, determining the amount of history to use itself is hard~\cite{peres2005two}, and selecting a large value of $m$ could severely overfit the data, thus making the learned model less reliable.

Strategies to resolve the above issues of higher-order Markov chains include variable order Markov chain~\cite{buhlmann1999variable} where the order length is a variable that can have different values for different states. There is a fitting algorithm that can automatically determine an appropriate order for each state, however it requires substantial computation time~\cite{ron1994learning} which restricts it to applications with only a small number of states~\cite{borges2007evaluating,deshpande2004selective,chierichetti2012web}.  Smoothing and regularization methods~\cite{chen1996empirical} like Kneser-Ney smoothing and Witten-Bell smoothing are additional approaches to make the higher-order Markov chain more robust. These methods are widely applied in language models for predicting unseen transitions. We will compare against the behavior of the Kneser-Ney smoothing in our experiments and show that our method has a number of advantages.

In this paper we propose the retrospective higher-order Markov process (RHOMP) as a simplified, special case of a higher-order Markov chain (Section~\ref{sec_method}). In this type of Markov model, a user retrospectively choses a state from the past $m$ steps of history, and then transitions as a first-order chain conditional on that state from history. This assumption helps to restrict the total number of parameters and protect the model from overfitting the correlations between history states. Specifically, this model corresponds to choosing $m$ different first order Markov chain transition matrices, one for each step of history,  as well as an associated probability distribution. Consequently, the number of parameters grows linearly with the size of history while preserving the higher-order nature. We also show there are important connections between our model and the class of pairwise-interaction tensor factorization models proposed by Rendle et al.~\cite{rendle2010pairwise,rendle2010factorizing} (Section~\ref{sec_pairwise}). 

We design an algorithm to select an optimal model from training data via  maximum likelihood estimation (MLE). For the second-order case with two steps of history, this yields a constrained convex optimization problem with a single hyperparameter $\alpha$. We derive a projected gradient descent~\cite{duchi2008efficient} algorithm to solve it. It requires only a few iterations to converge and each iteration is linear in the training data. We select the hyperparameter by fitting a polynomial to the likelihood function as a function of the parameter and select the global minimum. Thus, our RHOMP process does not require any parameter tuning and is scalable to applications with tens of thousands of
states. In addition, both the process of updating the gradients and model parameters parallelize over the training data.


We evaluate the effectiveness of RHOMP models in experiments ~\footnote{Code and data for this paper are available at: \url{https://github.
com/wutao27/RHOMP}.}
with real datasets including product reviews, online music streaming, photo locations, and checkin business types (Section~\ref{sec_data}).  We primarily compare algorithms in terms of their ability to predict information from testing data and use precision and mean reciprocal rank as the two main evaluation metrics. These experiments and results show that the RHOMP model achieves superior prediction results in all datasets (Section~\ref{sec_result}) compared with first and second order chains. For even higher-order chains, RHOMP shows stable performance with one exception (Section~\ref{sec_exp2}) where the data only has short sequences. 

\textbf{Remark.} Recently Kumar et.~\cite{kumar2017linear} proposed the Linear Additive Markov Process (LAMP) that is closely related to our framework. Specifically our RHOMP model has the same formulation as the generazlied extention GLAMP from the paper~\cite{kumar2017linear}. We learned about this paper as we were finalizing our submission to arXiv. The papers share a number of related technical results about the models and we discovered the related work~\cite{markov1971extension,zhang2015spatiotemporal,melnyk2006memory,usatenko2009random} based on their manuscript. The main difference is that in this paper we focus on the general form that allows to learn different Markov chains for each step of history. In addition we connect the RHOMP model with a particular tensor factorization to a higher-order Markov chain.

\section{PRELIMINARIES} 
\label{sec_prelim}

We begin by formally reviewing the problem of user trail prediction.
Then we will review relevant background on Markov chain models.

\subsection{Problem Formulation}
We denote a user trail as a sequence over a discrete state space
$s = (s_1, s_2, \cdots)$ with each element $s_i \in \{1,2,\cdots, N\}$. Here
$N$ is the total number of states.
The sequence can represent, for instance, a user's music listening history with each state denoting
a song/artist, or a user's checkin history from social network with each state denoting a location.
Given a specific user trail up to time $t-1$: $s = (s_1, s_2, \cdots, s_{t-1})$ with $t \geq 2$,
the task is to predict the next state at time $t$ based on a large set 
of user trails for training: $\mathcal{S} = \{s^{(1)}, s^{(2)},\cdots\}$, where each $s^{(i)}$ is
an individual trail.

\subsection{Markov Chain Methods}
An $m-$th order Markov chain is defined as a stochastic process
$\{ X_t, t = 1,2,\cdots \}$ on the state space: $\{1,2,\cdots, N\}$ with the property that the next transition only depends on the last $m$ steps. Formally, 
\[
\begin{split}
& \Pr\big(X_{t} = i \mid X_{t-1} = i_{t-1},\cdots,X_1 = i_1\big)\\
& \qquad = \Pr\big(X_{t} = i \mid X_{t-1} = i_{t-1},\cdots,X_{t-m} = i_{t-m}\big).
\end{split}
\]
An $(m+1)$-order transition tensor $\cmP$ with size $N$ characterizes the above Markov chain,
with $\cP_{i,j,\cdots,k}$ denoting the probability of transitioning to state $i$
given the $m$ current history states $(j, \cdots, k)$.
The model with $m=1$ is called the first-order Markov chain and similarly it can be described by an
$N \times N$ transition matrix $\mP$.


In order to use a Markov chain for the prediction problem, we need
to estimate the transition matrix $\mP$. 
Given a set of users trails $\mathcal{S} = \{s^{(1)}, s^{(2)},\cdots\}$, the maximum likelihood estimator (MLE)
of the probability $P_{i,j}$ for a first order chain is given by~\cite{chierichetti2012web}:
\[
P_{i,j} = \frac{c(i,j)}{\sum_{\ell} c(\ell,j)}
\]
where $c(i,j)$ denotes the number of instances that the states $j$ and $i$ were consecutive in all trails. For the case of higher-order Markov chain, it is well-known that any higher-order ($m>1$) Markov chain $X_t$ is equivalent to a first-order Markov chain $Z_t$ by taking a Cartesian product of its state space. This simplifies the parameter estimations and we may replace the original states with the Cartesian product states:
\[
\cP_{i,j,\cdots,k} = \frac{c(i,j,\cdots,k)}{\sum_{\ell} c(\ell,j,\cdots,k)},
\]
where now $c(i,j,\cdots,k)$ counts the number of instances of the sequence $k,\cdots,j,i$ in the training data.

Returning to the prediction task itself, Markov chain methods take as input the
history states of a trail and lookup the probabilities for all future states in the
matrix $\mP$ or tensor $\cmP$. This becomes 
a ranked list of states with the highest probability on top.


\section{Retrospective Higher-Order Markov Processes} 
\label{sec_method}

The goal of the retrospective higher-order Markov process (RHOMP)
is to strike a balance between the simplicity of the first order
Markov model and the high-parameter complexity of the higher-order
Markov model. Nevertheless, it is important for the model to 
account for higher-order behaviors because these are necessary to capture many types of user behaviors~\cite{rosvall2014memory,chierichetti2012web}. 
Towards that end, the RHOMP model describes a \emph{structured} higher-order Markov chain that results in a compact low-parameter description of possible user behaviors. We describe this formally for the case of a second-order history (and discuss largely notational extensions to higher-order chains in Section~\ref{sec_higher}). 


\subsection{The Retrospective Process} 
\label{sec_model}
The specific structure that a RHOMP describes is a retrospectively first-order Markov property. For some intuition, suppose that a web surfer had visited a search-query result page and then clicked the first link. In the RHOMP model, the user will first determine if they are going to continue browsing from the search-result page or the first link---hence users have the power to retrospect over history. Once that decision has been made, the user will behave in a first-order Markovian fashion that depends on \emph{if} the user returned to the previous state or remained on the current state.  Formally, suppose that the chain has recently visited states $j$ and $k$. The RHOMP is a two-stage process that \emph{first} selects a single history state. Since there are only two states, we model this selection as a weighted coin-toss where the probability of picking $j$ is $\alpha$ and so picking $k$ happens with probability $1-\alpha$. Once we have the history state, then the RHOMP transitions according to a transition matrix that is specific to that step of the history. Thus 
\[ \Pr\big(X_{t} = i \mid X_{t-1} = j,X_{t-2} = k\big) = \alpha R_{i,j} + (1-\alpha) Q_{i,k}, \]
where $\mR$ models the transitions from the current state (when those are selected) and $\mQ$ models the transitions from the previous state (when those are selected). See Figure~\ref{fig_illustration} for illustration. We summarize this in the following definition:
\begin{definition}
	Given $0\leq \alpha \leq 1$ and two stochastic matrices $\mR,\mQ$, a second-order retrospective higher-order Markov process 
	will transition from state $j$ with history state $k$ as follows: (i) with probability $\alpha$ it transitions 
	according to $\mR$ with the current state $j$, and
	(ii) with probability $1 - \alpha$ it transitions according to $\mQ$ with the previous state $k$.
\end{definition}

\begin{figure}[tb]
\includegraphics{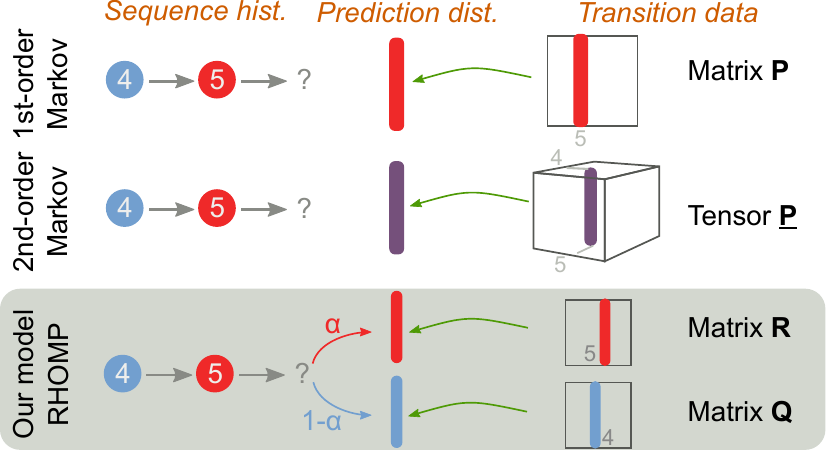}
\caption{An illustration of Markov chain methods and our proposed RHOMP model.}
\label{fig_illustration}
\end{figure}

This model has a number of useful features. For instance, it is easy to compute the stationary distribution as the following theorem shows.
\begin{theorem}
	Let $\alpha, \mR, \mQ$ be a second-order RHOMP model. Consider the stationary distribution $\vx$ in terms of the long-term fraction of time the process spends in a state: 
	\[ x_i = \lim_{t \to \infty} \frac{\text{number of times $X_t = i$}}{t}  \quad \text{ for each $i = 1 \ldots N$}. \]
	Such a distribution $\vx$ always exists. Moreover, it is unique if $\alpha \mR + (1-\alpha) \mQ$ is an irreducible matrix.   
\end{theorem}
\begin{proof} 
	Because the RHOMP is a special case of a second-order chain, we can use the relationship with the first-order chain on the Cartesian product space to establish that a distribution $\vx$ always exists. This follows because the long-term distribution of a first-order, finite-state space Markov chain always exists (though there could be multiple such distributions)~\cite{taylor2014introduction}. Let $X_{i,j}$ for all $1\leq i,j \leq N$ be any limiting distribution of the product state space, and $\vx$ be either of the corresponding marginal distribution such that
	$\sum_{j}X_{j,k} = x_{k}$ or $\sum_{k}X_{j,k} = x_{j}$. Note that both of these marginals result in the same distribution because 
	we use the long time average to define $X_{i,j}$.
    Then we have:
	\[
	\begin{split}
	x_{i} & = \sum_{j}X_{i,j}
	= \sum_{j}\sum_{k} (\alpha R_{i,j} + (1-\alpha) Q_{i,k})X_{j,k} \\
	& = \sum_{j} \alpha R_{i,j} x_{j} + \sum_{k} (1-\alpha)Q_{i,k}x_{k} = (\mP\vx)_{i}
	\end{split}
	\]
	where $\mP$ is defined as $\alpha \mR + (1- \alpha)\mQ $. So the limiting distribution
	$\vx$ follows $\vx = \mP \vx$, and it is unique if the corresponding Markov chain $\mP$ is
	irreducible.
\end{proof}

In Section~\ref{sec_opt}, we show how to compute a maximum likelihood estimate of $\mR$ and $\mQ$ from data.

\subsection{A Tensor Factorization Perspective}
\label{sec_pairwise} 
We originally derived this type of RHOMP via a tensor factorization approach, but then realized that the retrospective interpretation is more direct and helpful. Nevertheless, we believe there are fruitful connections established by the tensor 
factorization approach. Consider the transition tensor of a second-order Markov
chain:  $\cmP$ is a $3$-mode, $N \times N \times N$, non-negative tensor such that 
\begin{equation}
\label{eq_constraint}
\begin{split}
\sum_{i}\cP_{i,j,k} = 1 \text{ for all } 1 \leq j,k \leq N.
\end{split}
\end{equation}
This imposes a set of $N^2$ equality constraints. If we wanted to use traditional 
low-rank tensor approximations such 
as PARAFAC or Tucker~\cite{kolda2009tensor} to study large datasets, then we
would need to add a large number of constraints to the fitting algorithms in 
order to ensure that the factorization results in a stochastic tensor
that we could use for a second order Markov chain. This approach was extremely
challenging.

Instead, consider a pairwise interaction tensor factorization (PITF) as proposed by Rendle et al.~\cite{rendle2010pairwise} with the following form:
\begin{equation}
\label{eq_pair}
\cP_{i,j,k} = \sum_{\ell}A^{(J)}_{i,\ell} B^{(I)}_{j,\ell} + \sum_{\ell}A^{(K)}_{i,\ell} C^{(I)}_{k,\ell}
+ \sum_{\ell}B^{(K)}_{j,\ell} C^{(J)}_{k,\ell}
\end{equation}
where matrices $\mA^{(J)}, \mA^{(K)}, \mB^{(I)}, \mB^{(K)}, \mC^{(I)}, \mC^{(J)} \in \mathbb{R}^{N\times k}$.
We notice that last term in \eqref{eq_pair} is the interaction between the current state $j$ and the previous state $k$, and it contributes only a constant determined by the pair $(j,k)$. In the applications of prediction,
we can drop this term because it does not affect the relative ranking for the future state $i$. So the factorization
model becomes:
\begin{equation}
\label{eq_fac}
\cP_{i,j,k} = \sum_{\ell}A^{(J)}_{i,\ell} B^{}_{j,\ell} + \sum_{\ell}A^{(K)}_{i,\ell} C^{}_{k,\ell}
\end{equation}
with $\mA^{(J)}, \mA^{(K)}, \mB, \mC \in \mathbb{R}^{N\times k}$. 

To see the relationship with our RHOMPs, denote $\tilde{\alpha} \tilde{\mR} = \mA^{(J)}\mB^\intercal$ and $(1-\tilde{\alpha} )\tilde{\mQ} = \mA^{(K)}\mC^\intercal$ with
$0 \le \tilde{\alpha} \le 1$. Then the result of a PITF factorization with stochastic constraints is: 
\begin{equation}
\label{eq_model}
\cP_{i,j,k} = \tilde{\alpha} \tilde{R}_{i,j} + (1-\tilde{\alpha}) \tilde{Q}_{i,k}
\end{equation}
It is easy to verify that if both $\tilde{\mR}$ and $\tilde{\mQ}$ are stochastic matrices, then the corresponding
tensor $\cmP$ is a transition tensor following \eqref{eq_constraint}. The following
theorem shows that from any nonnegative $\tilde{\mR}$ and $\tilde{\mQ}$, we can construct such stochastic matrices.

\begin{theorem}
Assuming there exist nonnegative matrices $\tilde{\mR}$ and $\tilde{\mQ}$ such that the transition tensor $\cmP$ can
be decomposed in the form of \eqref{eq_model}, then there exist $0 \le \alpha  \le 1$ and 
stochastic matrices $\mR,\mQ$
such that
$\cP_{i,j,k} = \alpha R_{i,j} + (1-\alpha) Q_{i,k}$.
\end{theorem}

\begin{proof}
Denote $\sum_{i} \tilde{R}_{i,j} = \tilde{r}_j$ and $\sum_{i} \tilde{Q}_{i,k} = \tilde{q}_k$ for all $1\leq j,k \leq N$.
Because $1 = \sum_{i}\cP_{i,j,k} = \tilde{\alpha} \tilde{r}_j + (1-\tilde{\alpha}) \tilde{q}_k$ for all $1\leq j,k \leq N$,
we have $\tilde{r}_1 = \tilde{r}_2 = \cdots = \tilde{r}_N = \tilde{r} \geq 0$, $\tilde{q}_1 = \tilde{q}_2 = \cdots = \tilde{q}_N = \tilde{q} \geq 0$
and $\tilde{\alpha} \tilde{r} + (1-\tilde{\alpha})\tilde{q} = 1$. If $\tilde{r}=1,\tilde{q}=1$ then
the original matrices $\tilde{\mR}$ and $\tilde{\mQ}$ are stochastic. Otherwise we can set
\[
\alpha = \tilde{\alpha} \tilde{r}; \quad \mR = \tilde{\mR}/\tilde{r}; \quad \mQ = \tilde{\mQ}/\tilde{q}
\]
where $\mR$ and $\mQ$ are stochastic. Then we have 
\[
\begin{split}
\alpha R_{i,j} + (1-\alpha) Q_{i,k} & = \tilde{\alpha} \tilde{R}_{i,j} + \frac{(1-\tilde{\alpha} \tilde{r})\tilde{Q}_{i,k}}{\tilde{q}} \\
& = \tilde{\alpha} \tilde{R}_{i,j} + (1-\tilde{\alpha}) \tilde{Q}_{i,k} = \cP_{i,j,k}
\end{split}
\]
So $(\alpha,\mR,\mQ)$ forms a valid factorization for $\cmP$, the bound on $\alpha$ follows from $\tilde{\alpha} \tilde{r} + (1 - \tilde{\alpha}) \tilde{q} = 1$ from~\eqref{eq_model}.
\end{proof}
Consequently, the RHOMP form also arises from the PITF approach when constrained to model stochastic tensors.

\subsection{Parameter Optimization} 
\label{sec_opt}
In this section we will apply the principle of maximum likelihood to estimate the model parameters of a RHOMP (i.e., $\mR,\mQ$) directly from data. An alternative would be to estimate the higher-order Markov chain and use the PITF factorization as discussed in the previous section. Working directly on the RHOMP model from data has two advantages: first, the estimate corresponds exactly with the model, rather than estimate and approximate; and second, the direct approach is faster.

We first show how to compute a maximum likelihood estimate with $\alpha$ fixed and then discuss how to pick $\alpha$. 
Recall that $c(i,j,k)$ is the total count of transitions moving from $j$ to $i$ with previous state $k$ in the
training data.
With fixed $\alpha$, the log
likelihood of all transitions from the set $\mathcal{S}$ of user trails is:
\begin{equation}
\label{equ_likelihood}
\begin{split}
\log\mathcal{L}(\mR, \mQ \mid \mathcal{S}) &\quad = \quad \sum_{\mathclap{c(i,j,k)>0}} c(i,j,k) \log(\cP_{i,j,k}) \\
&\quad  = \quad \sum_{\mathclap{c(i,j,k)>0}} c(i,j,k) \log(\alpha R_{i,j} + (1-\alpha) Q_{i,k}) 
\end{split}
\end{equation}
Our goal is to find a pair of stochastic matrices $\mR,\mQ$ which
minimizes the negative log likelihood, which 
gives us the following optimization problem:
\begin{equation}
\label{eq_opt}
\MINtwo{\mR,\mQ}
{-\log \mathcal{L}(\mR, \mQ \mid \mathcal{S})}
{R_{i,j} \geq 0,\quad Q_{i,j} \geq 0 \quad 1 \leq i,j \leq N }
{\sum_{i}R_{i,j} = 1, \quad \sum_{i}Q_{i,k} = 1 \quad 1 \leq i \leq N}
\end{equation}
This optimization problem is convex as the following theorem shows. 
\begin{theorem}
The negation of the log likelihood function in~\eqref{equ_likelihood} is convex and so is
the feasible region of pairs of stochastic matrices. Thus any
local minima solution $(\mR^*, \mQ^*)$ is also the solution for global mimima.
\end{theorem}
\begin{proof}
First we verify the feasible domain of stochastic pairs $(\mR,\mQ)$ is convex. 
We can check that given $0\leq \lambda \leq 1 $ and two stochastic matrices $\mA,\mB$, the 
linear combination $\lambda \mA + (1-\lambda)\mB$ is also a stochastic matrix. This applies
element-wise to the pair to verify the claim. 

Now given two sets of stochastic matrices $(\mR^{(1)},\mQ^{(1)})$
and $(\mR^{(2)},\mQ^{(2)})$ and the corresponding linear combination
$(\mR = \lambda \mR^{(1)} + (1-\lambda)\mR^{(2)} ,\mQ = \lambda \mQ^{(1)} + (1-\lambda)\mQ^{(2)})$
we have
\[
\begin{split}
-  & \log\mathcal{L}(\mR, \mQ \mid \mathcal{S})  = 
- \sum_{i,j,k} c(i,j,k) \log(\alpha R_{i,j} + (1-\alpha) Q_{i,k}) \\
 & = - \sum_{i,j,k} c(i,j,k) \log\big( \lambda(\alpha R^{(1)}_{i,j} + (1-\alpha) Q^{(1)}_{i,k})\\
 &\qquad \qquad \qquad + (1-\lambda)(\alpha R^{(2)}_{i,j} + (1-\alpha) Q^{(2)}_{i,k}) \big) \\
 & \leq - \sum_{i,j,k} c(i,j,k) \big (\lambda \log (\alpha R^{(1)}_{i,j} + (1-\alpha) Q^{(1)}_{i,k})\\
 & \qquad \qquad \qquad + (1-\lambda)\log (\alpha R^{(2)}_{i,j} + (1-\alpha) Q^{(2)}_{i,k}) \big)\\
 & = - \lambda\log\mathcal{L}(\mR^{(1)}, \mQ^{(1)} \mid \mathcal{S}) -
 (1-\lambda) \log\mathcal{L}(\mR^{(2)}, \mQ^{(2)} \mid \mathcal{S})
\end{split}
\]
So~\eqref{eq_opt} is a convex problem. 
\end{proof}

We now derive the projected gradient descent algorithm for~\eqref{eq_opt}, which is summarized in Algorithm~\ref{alg_1}. 
This involves
\begin{enumerate}[leftmargin=0.3in]
	\item First update $\mR$ and $\mQ$ based on their gradients.
	\item Since $\mR$ and $\mQ$ are no longer stochastic due to the above updates, the projection
	step is applied to project the updated $\mR$ and $\mQ$ back to $\ell_1-balls$ (i.e., the stochastic property).
\end{enumerate}
The gradients over $\mR$ and $\mQ$ are:
\begin{equation}
\label{eq_gradient}
\begin{split}
\Delta R_{i,j} = \frac{-\partial \log \mathcal{L} }{\partial R_{i,j}} 
= \sum_{k}\frac{-\alpha c(i,j,k)}{\alpha R_{i,j} + (1-\alpha)Q_{i,k}}\\
\Delta Q_{i,k} = \frac{-\partial \log \mathcal{L} }{\partial Q_{i,k}}
= \sum_{j}\frac{-(1 - \alpha) c(i,j,k)}{\alpha R_{i,j} + (1-\alpha)Q_{i,k}}
\end{split}
\end{equation}
We accomplish the projection step using the algorithm from~\cite{duchi2008efficient}.
Note that for the sake of simplicity we present the projection step by sorting the vector
$\vw$, but there is a more efficient method based on divide and concur~\cite{duchi2008efficient} 
which is linear cost to the number non-zeros in $\vw$. However in practice sorting $\vw$ 
is fast as the vector $\vw$ is very sparse.

Overall each iteration takes linear time in the number of unique
triples $(i,j,k)$ in the sequence data. This is upper bounded by the size of input data. 
We also note that the procedure of computing the gradients
$\Delta \mR, \Delta \mQ$ and updating $\mR, \mQ$, which dominates the majority of the computation,
can be paralleled.



\begin{algorithm}[!htp]
\caption{Max.~Likelihood Estimate of a 2nd-order RHOMP}
\label{alg_1}
\begin{algorithmic}[1]
 \REQUIRE parameter $\alpha$, step size $\gamma_0$ and transition counts $c(i,j,k)$
\STATE Initialize $\mR$ with $R_{i,j} = \sum_{k}c(i,j,k)/\sum_{\ell,k}c(\ell,j,k)$,
$\mQ$ with $Q_{i,k} = \sum_{j}c(i,j,k)/\sum_{\ell,j}c(\ell,j,k)$ and $\gamma = \gamma_0$
\REPEAT
\STATE Compute the gradient matrices $\Delta \mR,\Delta \mQ$ based on ~\eqref{eq_gradient}
\STATE $\mR \leftarrow (\mR - \gamma \Delta \mR)$ and $\mQ \leftarrow (\mQ - \gamma \Delta \mQ)$
\FOR{each column vector $\vw$ of $\mR$ and $\mQ$}
\STATE Sort the non-zeros of $\vw$ into $\vu$: $u_1 \geq u_2 \geq \cdots \geq u_k > 0$
\STATE Find $\rho = \max \big\{ r \leq k: u_r - \frac{1}{r}(\sum_{i=1}^{r}u_i - 1) > 0\big\}$
\STATE Define $\theta = \frac{1}{\rho} (\sum_{i=1}^{\rho}u_i - 1)$
\STATE Update $\vw$ with $w_i \leftarrow \max\{ w_i - \theta, 0 \}$
\ENDFOR
\IF{objective value decreases}
\STATE $\gamma \leftarrow  \min\{ 2*\gamma, \gamma_0\}$
\ELSE
\STATE $\gamma \leftarrow 0.5*\gamma$; re-run this iteration with updated $\gamma$
\ENDIF
\UNTIL{converge}
\end{algorithmic}
\end{algorithm}

\noindent \textbf{Choosing $\bm{\alpha}$.} 
To determine the value of hyperparameter $\alpha$, we conduct a few trials with $\alpha$ 
chosen between $(0,1)$. Then based on the value of the objective function, we calculate the
best value of $\alpha$ from a polynomial interpolation of the likelihood function. 
Specifically $\alpha$ is selected
as $n$ Chebyshev nodes
$\alpha_k = \frac{1}{2} + \frac{1}{2}\cos(\frac{2k - 1}{2n}\pi),\ k = 1,2,\cdots,n$.
Getting the global minimum of a polynomial interpolant can be done efficiently, and polynomials
can approximate arbitrary continuous functions, which renders this a pragmatic choice. 
Another approach for selecting the value of $\alpha$ is to conduct cross validation with
grid search. However a different objective is needed as we could run into unseen transitions
in the validation set and the likelihood would go to $-\infty$. Alternatively we can use a
measurement like precision instead of likelihood. The main advantage of cross validation
is its ability to prevent overfitting. In our experiment we find this problem does not occur,
so we drop this procedure as it requires substantially more computation. 


\subsection{Higher-order Cases Beyond Second Order} 
\label{sec_higher}
The ideas discussed in the above sections also work for the higher-order cases
with $m\geq 3$. The RHOMP model becomes:
\[
\Pr(X_{t} = i \mid \scriptstyle X_{t-1} = j,X_{t-2} = k, \ldots, X_{t-m} = \ell ) \displaystyle 
= \alpha_1 R^{(1)}_{i,j} 
+\alpha_2 R^{(2)}_{i,k} + \cdots +\alpha_m R^{(m)}_{i,\ell} 
\]
where $0\leq \alpha_i \leq 1$ for $i=1,2,\cdots,m$, $\sum_{i}\alpha_i = 1$ and matrices
$\mR^{(i)}$ for $i=1,2,\cdots,m$ are stochastic. Similarly the log likelihood function
can be derived as well as the gradient over each $\mR^{(i)}$. The projected gradient descent 
algorithm is then applied to update each stochastic matrix $\mR^{(i)}$, with a
per-iteration complexity bounded by the size of the training data. 

The biggest difference
is that we are no longer able to determine the hyperparameters $\alpha_i$ in 
a simple fashion as the polynomial interpolation is only computationally efficient
for one or two parameters. To address this issue, recall that in Section~\ref{sec_model} 
we proposed the model as a retrospective walk, where the walker has probability $\alpha_k$ 
to step back ${k-1}$ steps into their history 
and then transition according to $\mR^{(k)}$. Our proposal is
to use a single hyperparameter $\beta< 1$ to model a decaying probability of looking back into the history:
\[ \alpha_1 = \tfrac{1-\beta^m}{1-\beta}, \quad \alpha_2 = \beta\tfrac{1-\beta^m}{1-\beta}, \quad \ldots, \quad \alpha_{m} = \beta^{m-1}\tfrac{1-\beta^m}{1-\beta}.\]
(This distribution describes a truncated geometric random variable.) 
In our experiments for the second-order case the optimal
$\alpha_1 > 1/2$ for every dataset. This offers a single step of evidence
for this assumption. 
This $\beta$ can be chosen either by the procedure of polynomial interpolation or simply 
using the optimal value $\alpha^*$ from a second-order factorization model
$\beta = \alpha^*/(1-\alpha^*)$. We apply the latter approach in our experiments for
RHOMP with $m>2$.



\section{Related Work} 
\label{sec_related}
 \textbf{Modeling User Trails.} Early work in~\cite{pirolli1999distributions} characterized the user path 
patterns on the web with the tools of Markov chains. Other advanced methods 
include hidden Markov models (HMM)~\cite{eddy1996hidden},
variable length Markov chains~\cite{buhlmann1999variable} and association rules~\cite{awad2012prediction}. 
However the computations associated with the above methods limit them from being used in datasets with
more than a few thousand states.
More recent work considers the sequence prediction task with personalization, such as collaborative filtering
methods~\cite{salakhutdinov2008bayesian,shi2014collaborative} where the behavior of similar users
is utilized to help the prediction, factorizing personalized Markov chains~\cite{rendle2010factorizing}, and
TribeFlow~\cite{figueiredo2016tribeflow}.
Other than the prediction problem, clustering and visualization~\cite{cadez2003model}, sequence classification~\cite{xing2010brief}, metric embedding~\cite{chen2012playlist,chen2013multi}
and hypotheses comparison~\cite{singer2015hyptrails} have also been studied. In the context of this work,
we seek to improve the performance of the classic and simple Markov model by studying a structured variation. 
 
\textbf{Random Walk Models.} Since our model is a special case of a higher-order Markov chain,
we note that there are relationships with a variety of enhanced Markov models.
First our RHOMP model defines a specific form of the Additive Markov Process (AMP)~\cite{markov1971extension}, where the transition probability is a summation of a series of memory functions that are restricted on the next state and one history state each. Applications of the AMP include LAMP~\cite{kumar2017linear} (see Section~\ref{sec_intro}), the gravity models~\cite{zhang2015spatiotemporal}, and some dynamical systems in physics~\cite{melnyk2006memory,usatenko2009random} where the memory function is empirically estimated for the application of binary state.
In addition to the AMP, recent innovations include 
new recovery results on mixture of Markov chains~\cite{gupta2016mixtures} (a special case of HMM), which 
assumes a small set of Markov chains that model various classes of latent indent;
and the spacey random walk~\cite{benson2016spacey,benson2015tensor,wu2016general} as a non-Markovian stochastic 
	process that utilizes higher-order information based on the empirical occupation of states.

\textbf{Tensor Factorization.} As already discussed, our work is directly related to the pairwise 
interaction tensor factorization (PITF) method proposed
by Rendle in~\cite{rendle2010pairwise,rendle2010factorizing}, where the task is to generate \textit{tag}
recommendations given the $\{\textit{user, item}\}$ combination. The PITF model is learned from a binary
tensor of triple $\{\textit{user, item, tag}\}$ by bootstrap sampling from pairwise ranking constrains.
Our work differs in the aspect of problem formulation, model construction and parameter optimization.
The RHOMP model is also a special case of both the canonical/PARAFAC and Tucker decompositions~\cite{kolda2009tensor}.


\section{Experiments} 
\label{sec_exp}
\begin{table}[b]
\vspace*{-3ex}
\caption{Dataset characteristics in terms of the number of states, transitions and trails}
\label{table_dataset}
\begin{center}
\vspace*{-1ex}
\begin{tabularx}{0.95\linewidth}{Z Z Z Z}
\toprule
 & \# states & \# transitions & \# trails \\
\cmidrule{2-4}
LastFM & 17,341 & 2,902,035 & 195,499  \\
BeerAdvocate & 2,324 & 1,348,903 & 35,629  \\
BrightKite & 11,465 & 400,340 & 125,437  \\
Flickr & 7,608 & 1,212,674 & 97,563  \\
FourSQ & 344 & 198,503 & 1,480  \\
\bottomrule
\end{tabularx}
\end{center}
\end{table}
We evaluate our RHOMP method on the ability to predict subsequent states in a user trail in terms of precision and mean reciprocal rank (MRR) on five different types of data (Section~\ref{sec_data}). We then present the
results of a second-order
(i.e., $m=2$) RHOMP compared with baseline methods in Section~\ref{sec_result} and study over-fitting of the training data in Section~\ref{sec_ana}.
Then we study what happens for higher-order (i.e., $m>2$) models in Section~\ref{sec_exp2}.  In all cases, the RHOMP model offers a considerable improvement 
to existing methods.

\begin{table*}[!htb]
\vspace*{-1ex}
\caption{Mean Reciprocal Rank (MRR) results of various methods on all datasets.
Bold indicates the best mean performance, and $\pm$ entries
are the standard deviations over 5 trials.
Our proposed RHOMP ($m=2$) has the best performance in all datasets.}
\label{table_mrr}
\begin{center}
\begin{tabularx}{1.0\linewidth}{Z Z Z Z Z Z Z Z}
\toprule
 & MC1 & MC2 & Kneser1 & Kneser2 & PITF & LME & RHOMP \\
\cmidrule{2-8}
LastFM & $0.071 \pm 0.001$ & $0.068\pm 0.001$ & $0.066 \pm 0.001$ & $0.090 \pm 0.002$ & $ 0.058 \pm 0.001 $ & $0.062 \pm 0.001$ & $\bm{0.100} \pm  0.001$\\
BeerAdvocate & $0.080 \pm 0.000$ & $0.034\pm 0.001$ & $0.079 \pm 0.000$ & $0.076 \pm 0.001$ & $0.067 \pm 0.002$  & $0.067 \pm 0.001$ & $\bm{0.090} \pm  0.000$\\
BrightKite & $0.551 \pm 0.002$ & $0.540\pm 0.002$ & $0.554 \pm 0.002$ & $0.599 \pm 0.002$ & $ 0.440 \pm 0.007 $ & $0.529 \pm 0.002$ & $\bm{0.603} \pm  0.002$\\
Flickr & $0.358 \pm 0.003$ & $0.306\pm 0.004$ & $0.350 \pm 0.001$ & $0.379 \pm 0.001$ & $ 0.313 \pm 0.004 $ & $0.333 \pm 0.003$ & $\bm{0.410} \pm  0.001$\\
FourSQ & $0.138 \pm 0.004$ & $0.092\pm 0.003$ & $0.146 \pm 0.005$ & $0.155 \pm 0.004$ & $0.120 \pm 0.003 $ & $0.113 \pm 0.002$ & $\bm{0.181} \pm  0.003$\\
\bottomrule
\end{tabularx}
\end{center}
\end{table*}

\begin{figure*}[!htp]
\vspace*{-1ex}
  \centering
  \includegraphics[trim={1.3cm 5.2cm 1.5cm 5.2cm}, width=.3\linewidth]{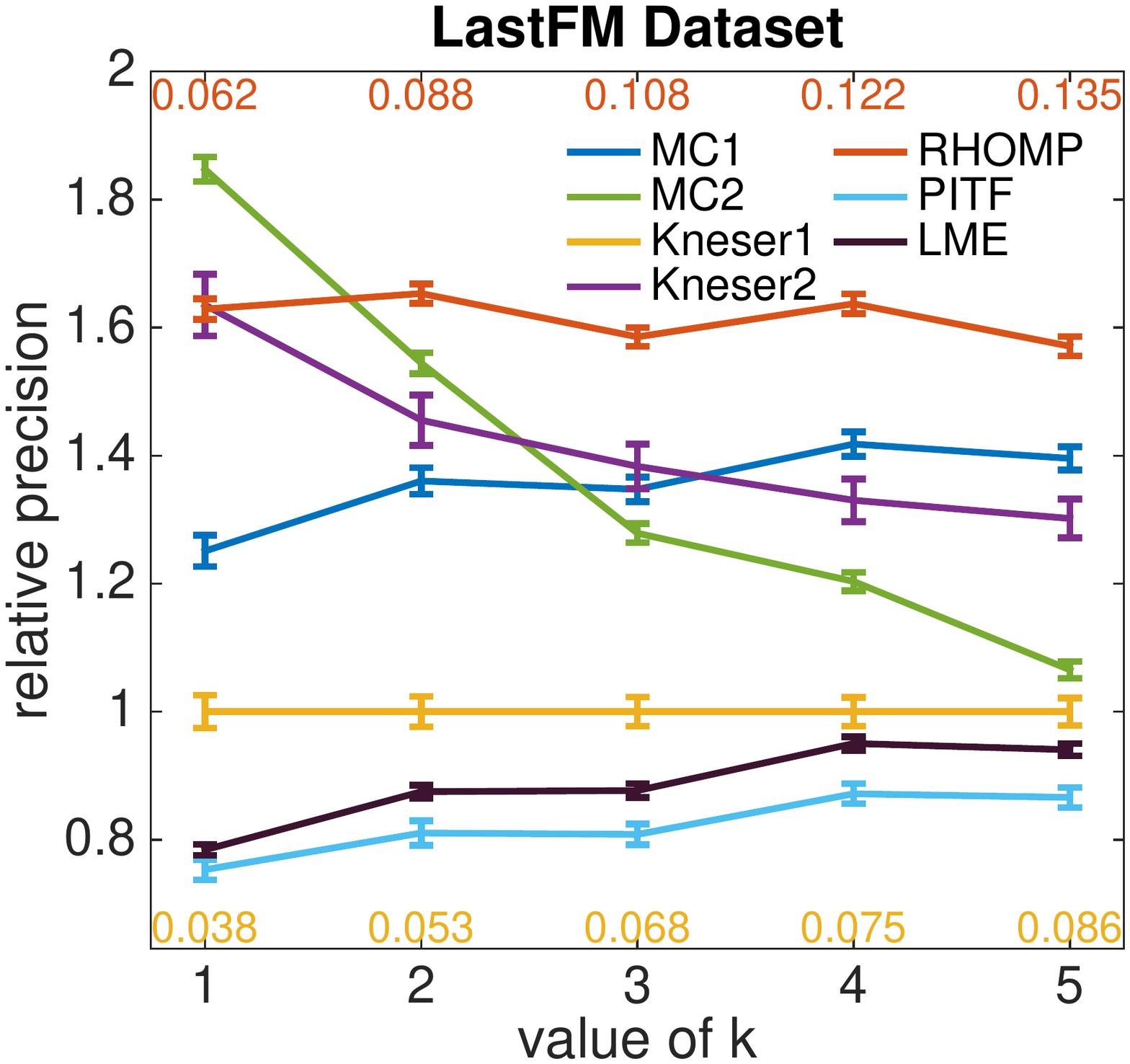}
  \includegraphics[trim={1.3cm 5.2cm 1.5cm 5.2cm}, width=.3\linewidth]{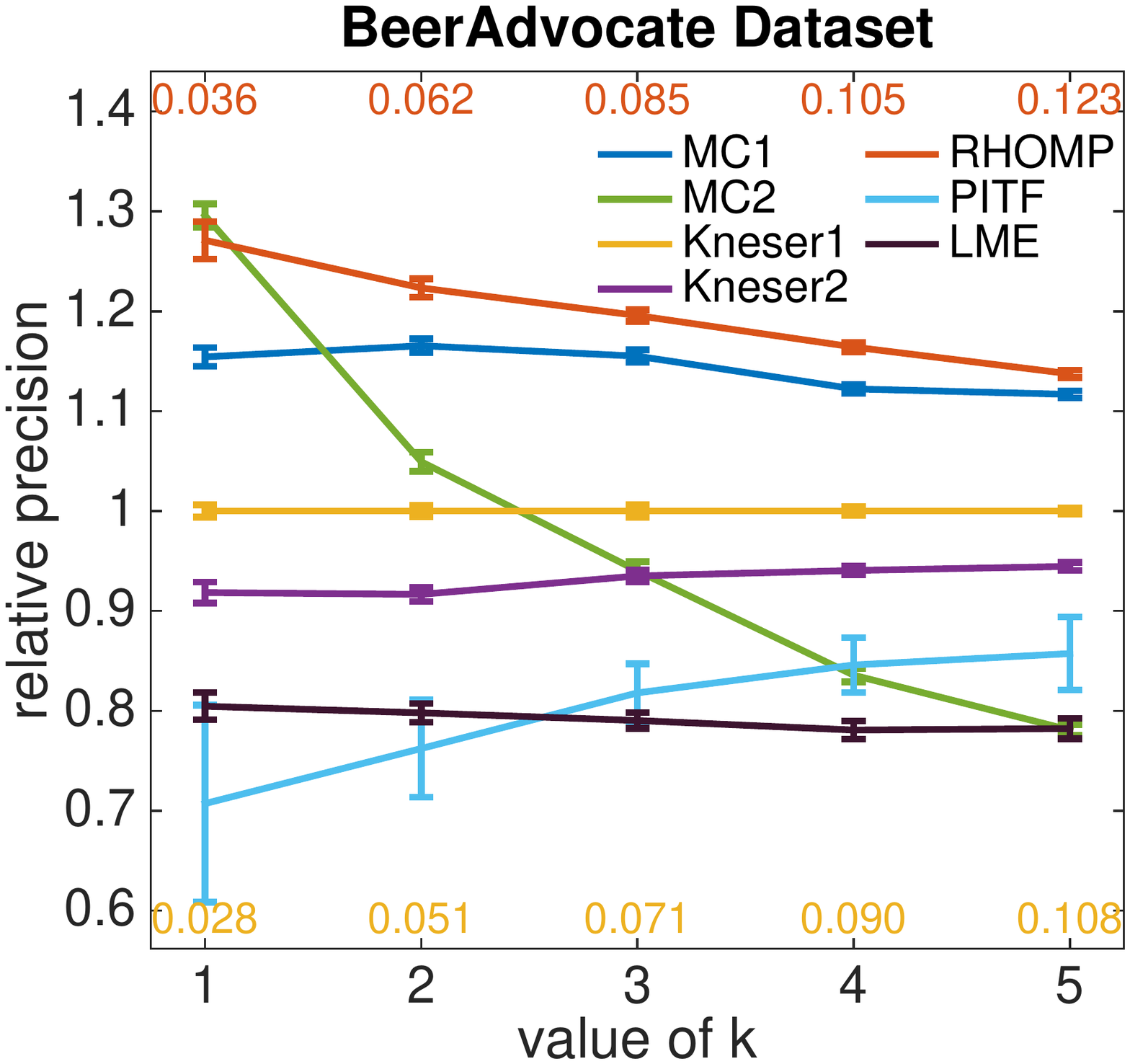}
  \includegraphics[trim={1.3cm 5.2cm 1.5cm 5.2cm}, width=.3\linewidth]{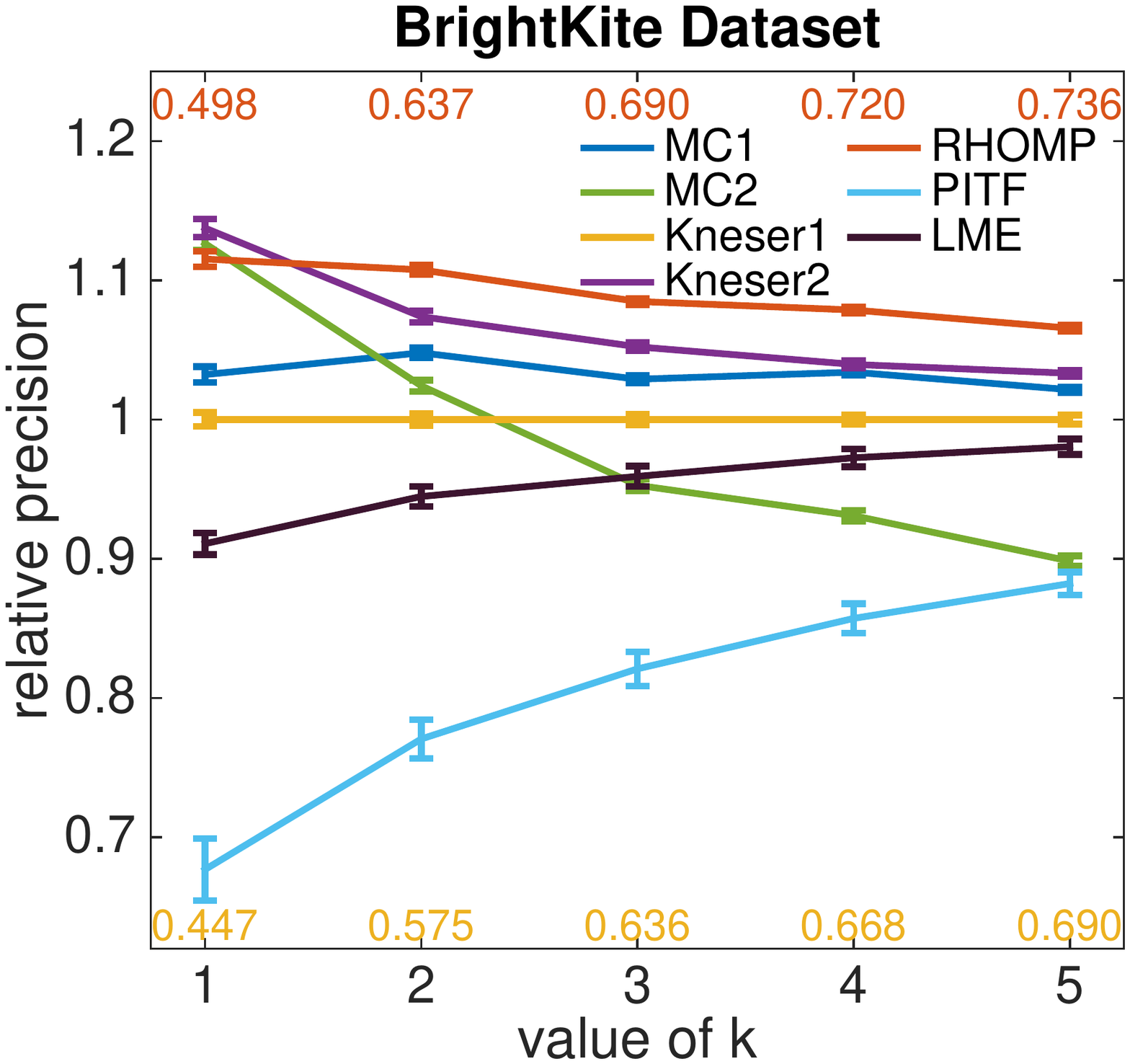}
  \includegraphics[trim={1.3cm 5.2cm 1.5cm 5.2cm}, width=.3\linewidth]{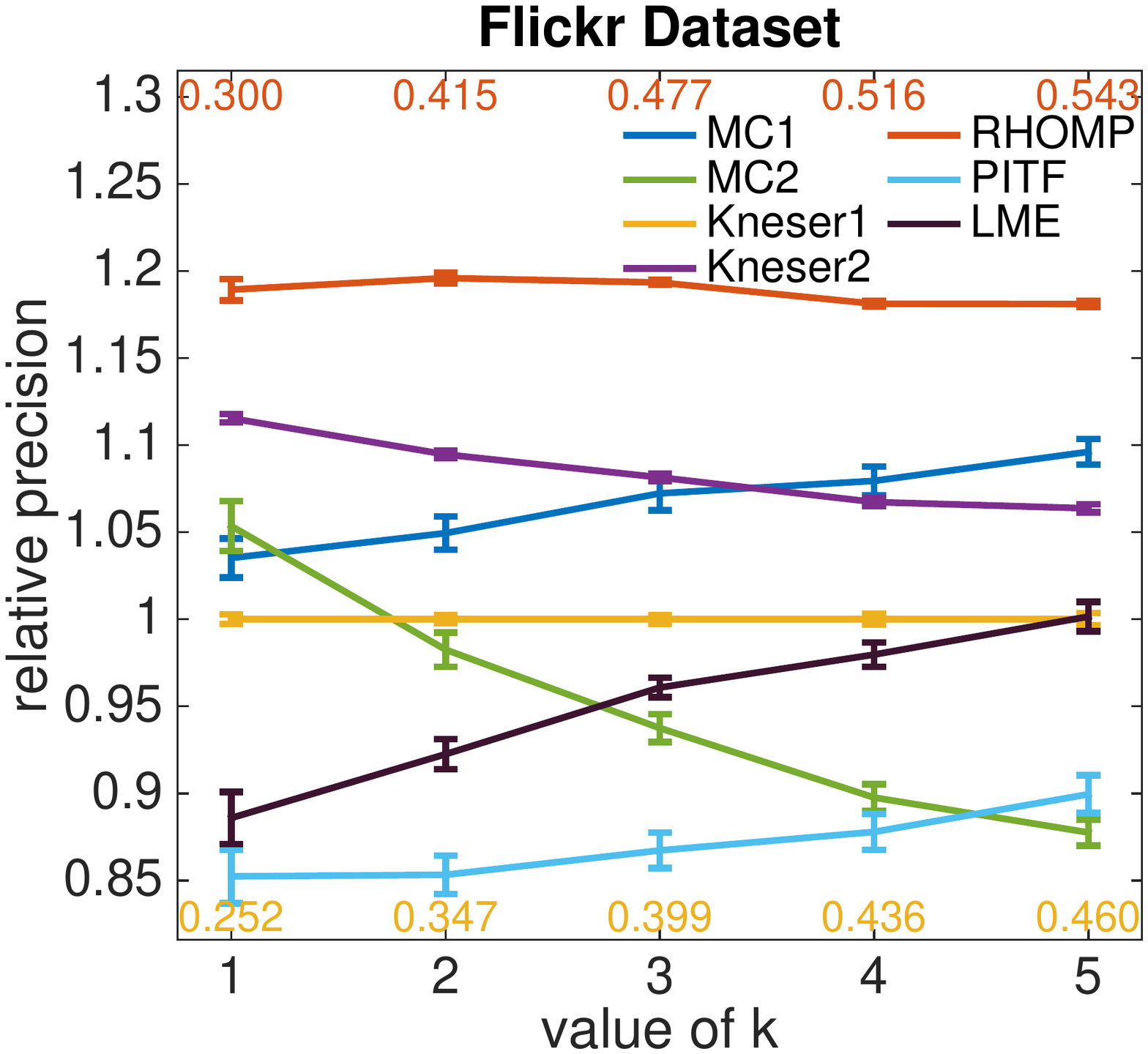}
  \includegraphics[trim={1.3cm 5.2cm 1.5cm 5.2cm}, width=.3\linewidth]{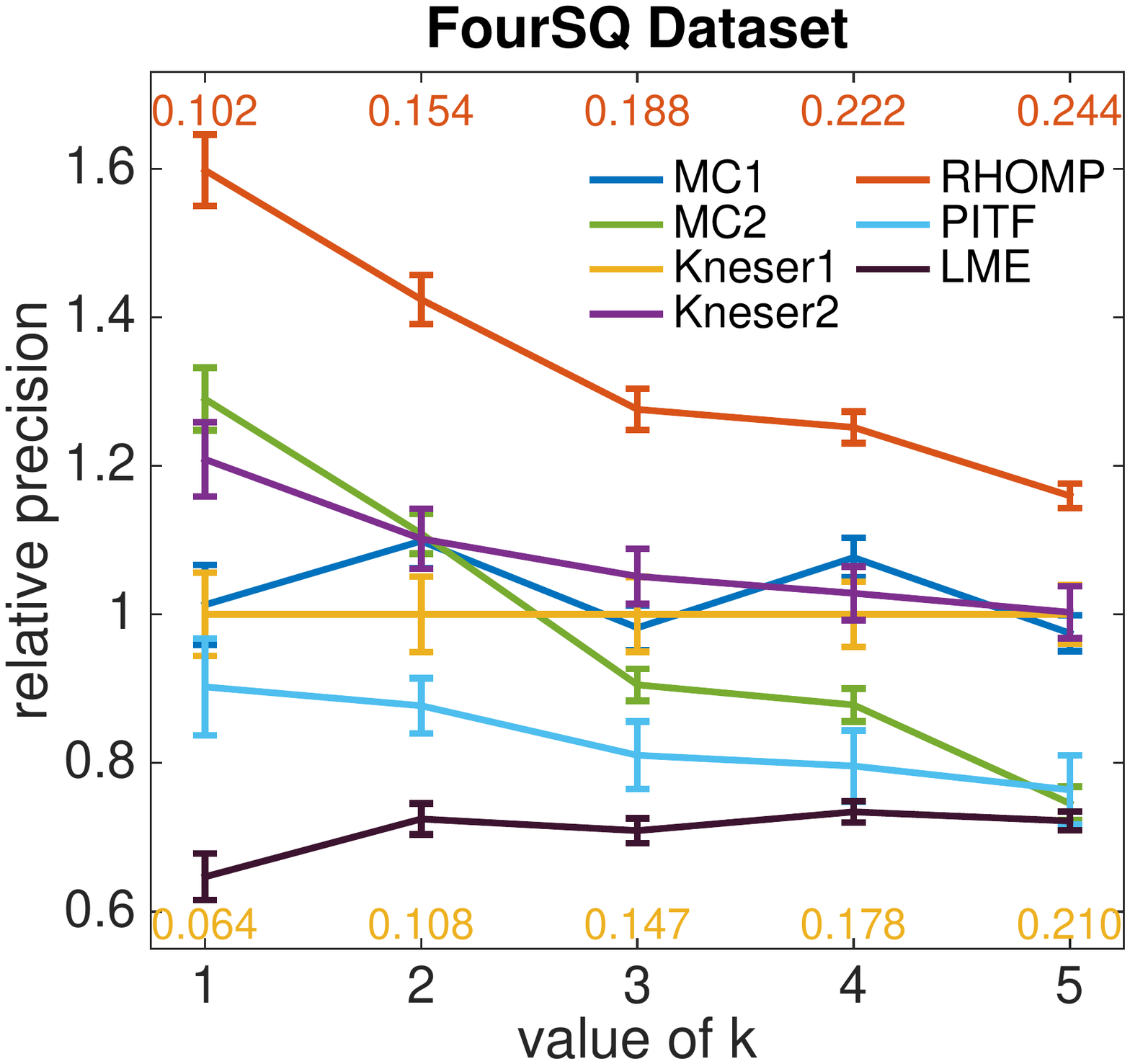}
  \caption{Relative precision results on all datasets with $k = 1,2,3,4,5$. We use Kneser1 as the baseline, and
  the relative precision is calculated as the precision ratio to that of Kneser1. The error bars in
  the figure are the standard deviations over 5 trials. The numbers in the bottom and the top
  of the figures denote
  the absolute precisions for the Kneser1 and our RHOMP method respectively.
   We see that our RHOMP has noticeable
  improvements over other methods in most datasets.}
  \label{fig_precision}
  \vspace*{-1ex}
\end{figure*} 


\subsection{Datasets and Evaluations Setup}
\label{sec_data}
The real datasets we use in our experiments cover several applications including:
product reviews, online music streaming, checkin locations of social network and photo uploads.
Every dataset is publicly available.
For all the datasets self-loops are removed
as we are mostly interested in predicting a non-trivial transition. Also we only consider states
that show up more than $20$ times. Simple statistics on each dataset are summarized in Table~\ref{table_dataset}, and we now describe them individually.


\noindent \textbf{LastFM}~\cite{celma2009music}
is a music streaming and recommendation website (\url{last.fm}).
We generate user trails as listening histories regarding different artists
over a one-year period (2008-05-01 to 2009-05-01).

\noindent \textbf{BeerAdvocate}~\cite{mcauley2013amateurs} consists of beer reviews
spanning more than 10 years up to November 2011 from \url{beeradvocate.com}.
We study the user trail as reviews over different brewers.

\noindent \textbf{BrightKite}~\cite{cho2011friendship} was a location-based social networking website
where users shared their locations by checking-in. We study the trails of location id.

\noindent \textbf{Flickr}~\cite{thomee2015new} contains 100 million Flickr photos/videos
provided by Yahoo! Webscope.
We extract the user trail based on geolocation (restricted to USA) of each upload after 2008-01-01.
Each longitude and latitude is mapped into a grid of approximate 10km by 10km, which constitutes the state. 

\noindent \textbf{FourSQ} is a location based check-in dataset created by Yang et al~\cite{yang2013fine} which
contains checkins from New York City from 24 October 2011 to 20 February 2012. We extract checkin place
category (e.g., bus station, hotel, bank) as state.

For experimental methods, we consider the following:

\noindent \textbf{MC1, MC2} are the first-order and second-order Markov chain methods respectively, where
the transition matrix is estimated based on maximum likelihood.

\noindent \textbf{Kneser1, Kneser2} are the interpolated Kneser-Ney smoothing methods~\cite{chen1996empirical} applied on the first-order
and second-order Markov chain methods respectively. This
is one of the best smoothing methods for n-gram language models, where it enables higher-order Markov chain
transitions to unseen n-grams. We set the discounting parameter as $n_1/(n_1+2n_2)$ by the method of leaving one out~\cite{chen1996empirical}, where $n_1$ and $n_2$ denote the number of n-grams that appear exactly once and twice respectively

\noindent \textbf{PITF} is the pairwise interaction tensor factorization method~\cite{rendle2010pairwise} computed on the higher-order Markov chain
estimate. Because we use ranking, we consider general positive and negative
entries as valid for the factorization. We implement the fitting method
ourselves to handle the sparsity in our data.
As suggested in the paper~\cite{rendle2010pairwise}, the hyperparameters are $\lambda = 5\cdot10^{-5}$ and $\alpha = 0.05$ with initialization from
$N(0,0.01)$. We set the rank number $k$ as $5\%$ of the total number of states, which is enough to accurately capture the user behavior~\cite{rendle2010pairwise}. The number of iterations for the stochastic gradient descent is 10,000,000.

\noindent \textbf{LME} is short for Latent Markov Embedding~\cite{chen2012playlist}. It is an machine learning
algorithm that embeds states into Euclidean space based on a regularized maximum likelihood principle. We set the dimension $d = 50$ and
use default values all other parameters (e.g., learning rate, epsilon).
(We tried various values of $d$ spanning from $2$ to $100$, we find as $d$ increases
the performance also gets better, for $d>50$ the improvements are negligible. So we use $d=50$
to make the algorithm efficient.) We use the authors' implementations. 

\noindent \textbf{RHOMP} is our proposed method in this paper. We use initial step size as
$\gamma_0 = 1$, and set $\epsilon = 10^{-5}$ as the algorithm termination criterion when
the relative improvement over log likelihood is below this point. For the hyperparameter $\alpha$
we use $n=15$ Chebyshev nodes for the interpolation.


The datasets are randomly split into a training set ($60\%$) and testing set ($40\%$) based on keeping whole trails together. And for each
dataset we conduct experiments over $5$ random repetitions and present the average results.
For evaluations we use precision over top $k$ outputs to measure the accuracy of each method. It is calculated over all individual transitions in the testing set
as
\[\text{Precision}_k = \frac{\text{\# true transitions within top $k$ algorithmic results}}{\text{\# total transitions}}.\]
Besides precision, which measures the accuracy of the top outputs from algorithms, we also provide results on Mean Reciprocal Rank (MRR). The reciprocal rank of an output is the inverse of the rank of the ground truth answer and MRR measures the overall ranking compared to the groundtruth.  
For both measures, we want large scores close to 1. 

\subsection{General Results}
\label{sec_result}
First we compare our RHOMP ($m=2$) with other baseline methods in terms of precision and MRR score.

\noindent \textbf{MRR score.}
Table~\ref{table_mrr} depicts the results on the MRR score.
In all datasets, RHOMP has the highest score. From the table we see that MC1 outperforms
the LME method. The LME has the advantage of embedding the states into Euclidean space for
applications like visualization or clustering. However the embedding could potentially cause the information
loss, thus make the prediction less accurate.
 And we notice that MC2 has
the lowest scores in many cases (i.e., BeerAdvocate, Flickr and FourSQ datasets), and
the MRR scores drop compared to MC1. The Kneser-Ney smoothing modification makes the MC2 estimate more
robust, and in most cases outperforms the MC1, although such advantage is limited compared to that from our RHOMP method. The PITF method is also not competitive.

\noindent \textbf{Precision score.} Figure~\ref{fig_precision} shows the algorithms performances in terms of relative precision.
Many of the observations from Table~\ref{table_mrr} on the MRR score also apply here.
In addition we find MC2
is often able to provide one accurate output, so the relative precision ($k=1$) is actually quite good
in most cases. However as $k$ increase the relative precision drops rapidly
due to the fact that MC2 is not able to generate a few more reliable outputs.
This limits the application of MC2 because in the task of recommendation, it is important for the algorithm to
generate a few instead of one unique candidate state. 
Another observation is that the results of PITF over different trials are often more volatile
because of its underlying stochastic gradient descent solver.
We also find that for some datasets (e.g., BeerAdvocate and FourSQ) the relative precisions of our RHOMP decrease
as $k$ increases. The reason is that as $k$ increases, the prediction task itself becomes easier, so it is hard
to maintain the same advantage (i.e., constant relative precision). Same reason for the fact the inferior methods
like LME and PITF can catch up as $k$ increases.

\noindent \textbf{Algorithm Runtime.}
Table~\ref{table_runtime} shows the runtime for each method. The RHOMP approach takes slightly more time to train than Kneser-Ney methods, but has faster prediction and testing. It is slower than the pure MC methods, but much faster than PITF, LME. 

\begin{table}[!tb]
\caption{Algorithm runtime (in minutes) for the three large datasets in terms of training time (left) and testing time (right). The experiments are run on a single-core of a 2.5Ghz Xeon CPU. Both MC1 and MC2 ran in under a minute.}
\label{table_runtime}
\centering
\begin{tabularx}{1.0\linewidth}{Z H H Z Z Z Z Z}
\toprule
 & MC1 & MC2 & Kneser1 & Kneser2 & PITF & LME & RHOMP \\
\cmidrule{2-8}
LastFM & $ 1 /1 $ & $ 1/- $ & $ 2/4 $ & $3/75$ & $493/1980$ & $3188/57$ & $52/2$\\
BrightKite & $ -/- $ & $ -/- $ & $ <\!\!1/1 $ & $<\!\!1/4$ & $236/71$ & $1153/22$ & $3/1$\\
Flickr & $ -/- $ & $ -/- $ & $ <\!\!1/1 $ & $1/8$ & $168/97$ & $764/11$ & $6/1$\\
\bottomrule
\end{tabularx}
\end{table}

\subsection{Analysis on Overfitting}
\label{sec_ana}
\begin{table*}[!htb]
\caption{Precision (k=3) results for testing set (the left number) vs train set (the right number) that we use to estimate overfitting. 
Bold denotes the highest testing result.
We judge the overfitting effects as  $\{$MC2, Kneser2$\}$ $\gg$ $\{$MC1, Kneser1, RHOMP$\}$ $>$ $\{$LME, PITF$\}$. But LME and PITF have poor test performances.
}
\vspace*{-1ex}
\label{table_tt}
\begin{center}
\begin{tabularx}{1.0\linewidth}{Z Z Z Z Z Z Z Z}
\toprule
 & MC1 & MC2 & Kneser1 & Kneser2 & PITF & LME & RHOMP \\
\cmidrule{2-8}
LastFM & $0.092 / 0.216 $ & $0.087 / 0.961 $ & $0.068 / 0.109 $ & $0.094 / 0.792 $ & $ 0.055 / 0.061 $ & $0.060 / 0.083 $ & $\bm{0.108} / 0.218 $\\
BeerAdvocate & $0.082 / 0.115 $ & $0.067 / 0.777 $ & $0.071 / 0.074 $ & $0.066 / 0.490$ & $ 0.058 / 0.059 $  & $0.056 / 0.600 $ & $\bm{0.085} / 0.109 $\\
BrightKite & $0.654 / 0.782 $ & $0.606 / 0.940 $ & $0.636 / 0.729 $ & $0.669 / 0.868 $ & $ 0.522 / 0.575$ & $0.610 / 0.665 $ & $\bm{0.690} / 0.796$\\
Flickr & $0.428 / 0.496 $ & $0.374 / 0.832 $ & $0.399 / 0.440 $ & $0.432 / 0.710$ & $ 0.346 / 0.359 $ & $0.384 / 0.401 $ & $\bm{0.477} / 0.530 $\\
FourSQ & $0.145 / 0.199 $ & $0.133 / 0.778 $ & $0.147 / 0.174$ & $0.155 / 0.524 $ & $ 0.119 / 0.126 $ & $0.104 /0.137 $ & $\bm{0.188} / 0.241$\\
\bottomrule
\end{tabularx}
\end{center}
\end{table*}

One of the reasons we propose the RHOMP method is to improve the higher-order
Markov chain method in the aspect of overfitting. In this section we analyze
the results in detail and give an explanation on the performances of different methods.

First we show the comparison between training and testing performance in Table~\ref{table_tt}.  We present the result using precision with $k=3$ as it is
representative of the remaining results. Both PITF and LME had the least overfitting effect as the testing and training precisions are very close. However, their testing
precisions are also low. 
The training precision of MC2 is the
highest for all datasets. But these are often more than $10$ times of 
the corresponding testing precisions. So MC2 is a highly overfitting method.
Kneser2 also has
comparatively high training precision since it is a second-order method and tends to fit
the training data well. But the performance on testing set is better than MC2 as it uses lower-order information to smooth the output. The methods MC1, Kneser1 and RHOMP have a good training and testing balance, and among them, our RHOMP has superior testing performances.

Next we analyze the performance on individual states to help understand the behaviors of different
algorithms. We sort all the states from high to low based on the total number of counts of each state in
the training set. Our aim is to look at how testing accuracy correlates with these counts. 
Figure~\ref{fig_detail} shows the precision ($k=3$) comparisons (i.e., MC1 vs MC2 vs RHOMP and
Kneser1 vs Kneser2 vs RHOMP) on the Flickr dataset based on counts of the states.  We aggregate small sets of states based on their counts into baskets of at least 1000 transitions and 5 states. 
We find that all methods show precision drops when predicting infrequent states, with MC2 being affected most. Here, RHOMP does the best out of all methods, which reflects its ability to avoid overfitting.

\begin{figure}[tb]
  \centering
  \includegraphics[trim={1.5cm 5.2cm 2.5cm 4.8cm}, width=.49\linewidth]{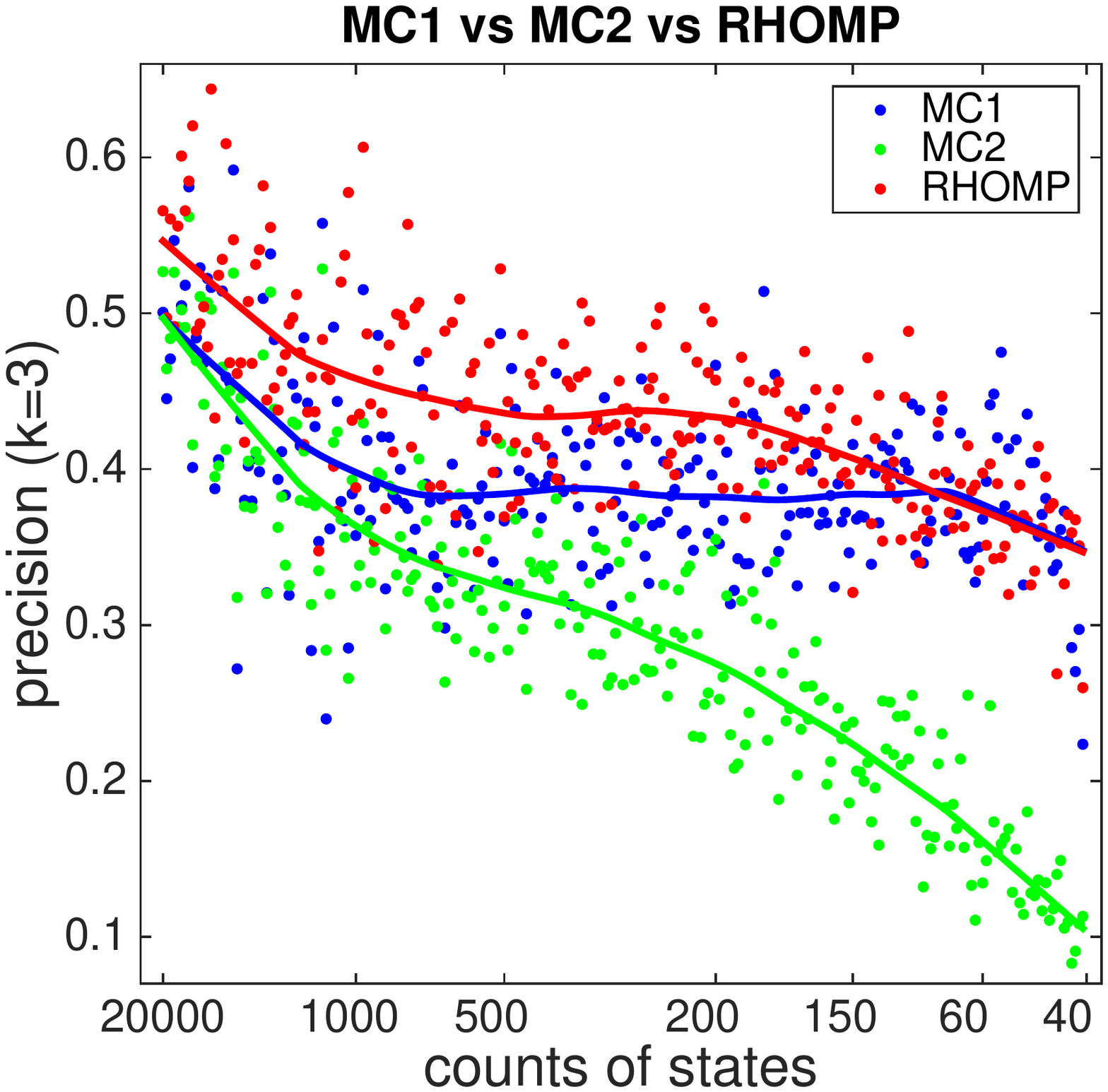}
  \includegraphics[trim={1.5cm 5.2cm 2.5cm 4.8cm}, width=.49\linewidth]{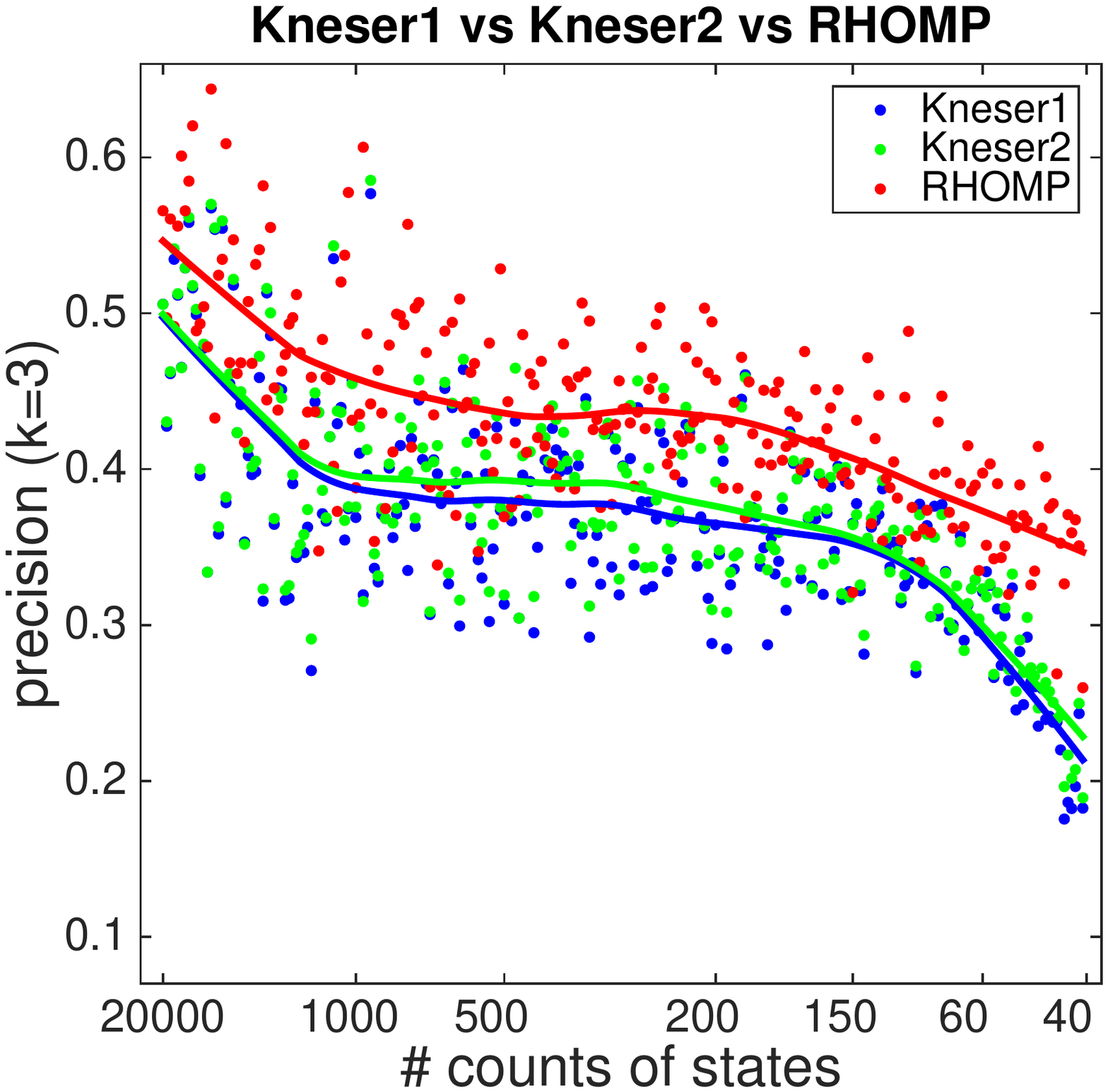}
  \vspace*{-1ex}
  \caption{State-wise precision ($k=3$) comparison on MC1 vs MC2 vs RHOMP (left figure)
  and Kneser1 vs Kneser2 vs RHOMP (right figure) on the Flickr dataset.
  Each marker represents the average precision over a group of states.
 The curves are fit from the scatter points
  based on Locally Weighted Scatterplot Smoothing (LOWESS).}
  \vspace*{-1ex}
  \label{fig_detail}
\end{figure}

\begin{figure*}[!htp]
  \centering
  \includegraphics[trim={1.6cm 5.2cm 1.6cm 4.8cm}, width=.195\linewidth]{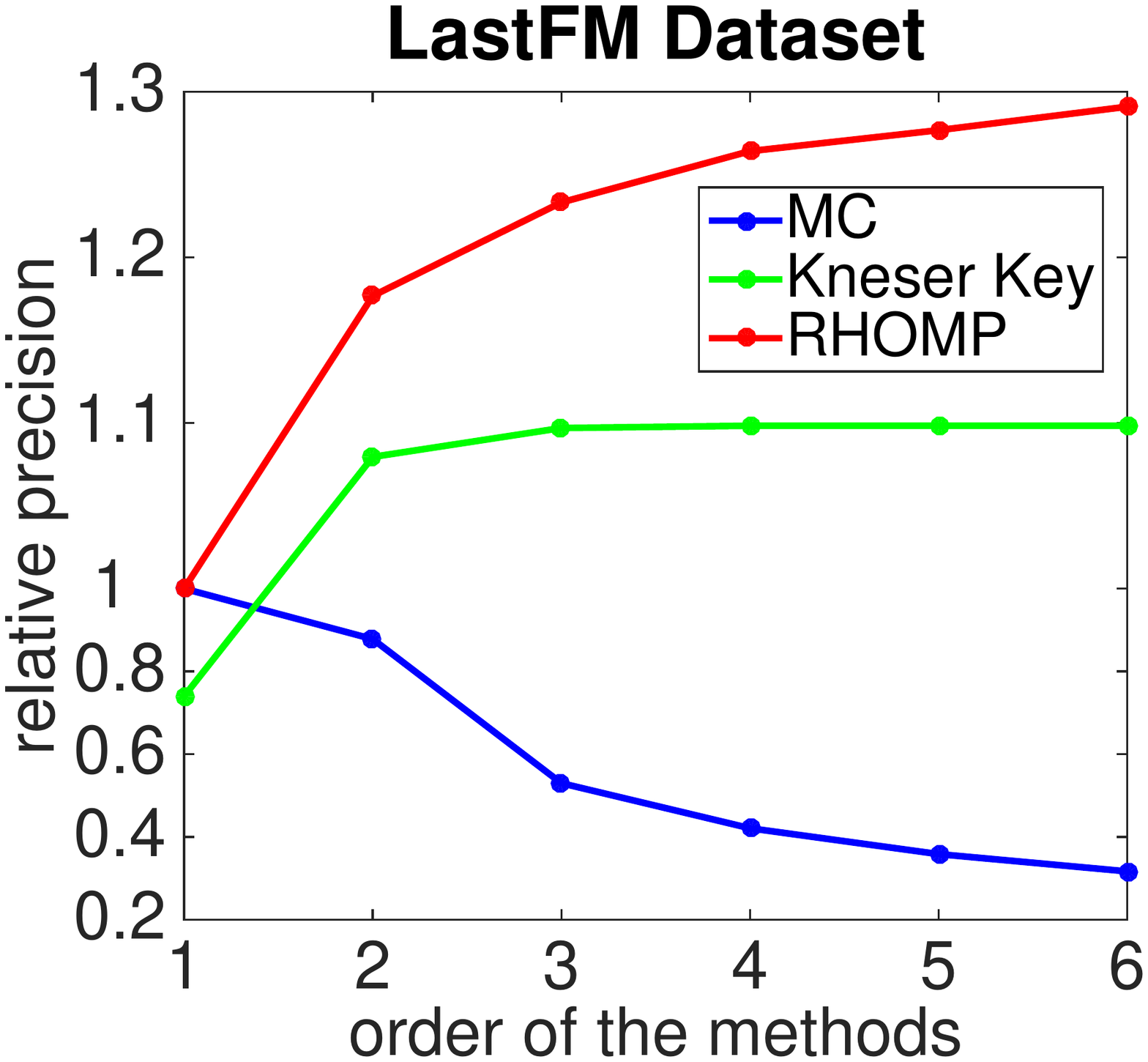}
  \includegraphics[trim={1.6cm 5.2cm 1.6cm 4.8cm}, width=.195\linewidth]{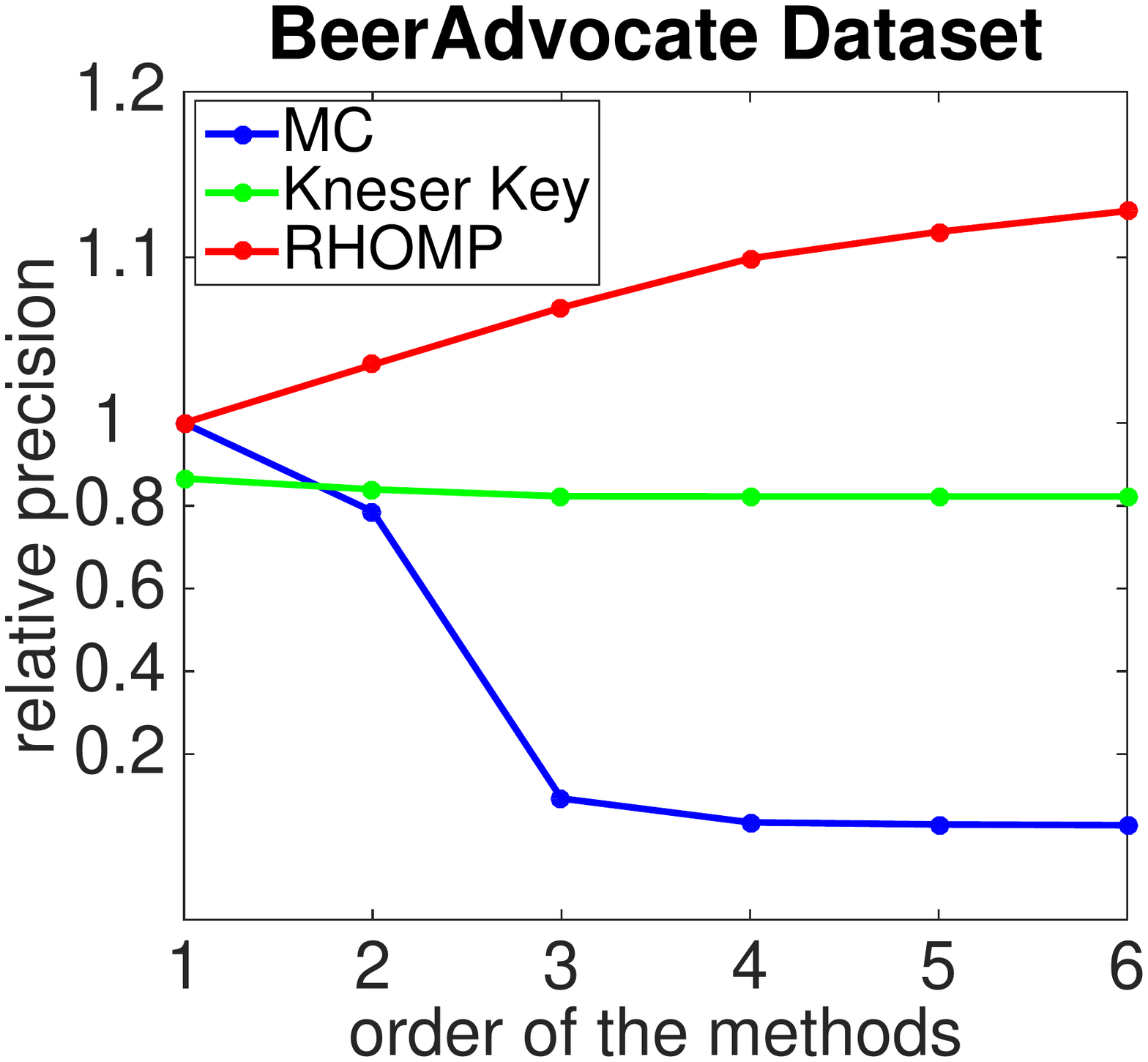}
  \includegraphics[trim={1.6cm 5.2cm 1.6cm 4.8cm}, width=.195\linewidth]{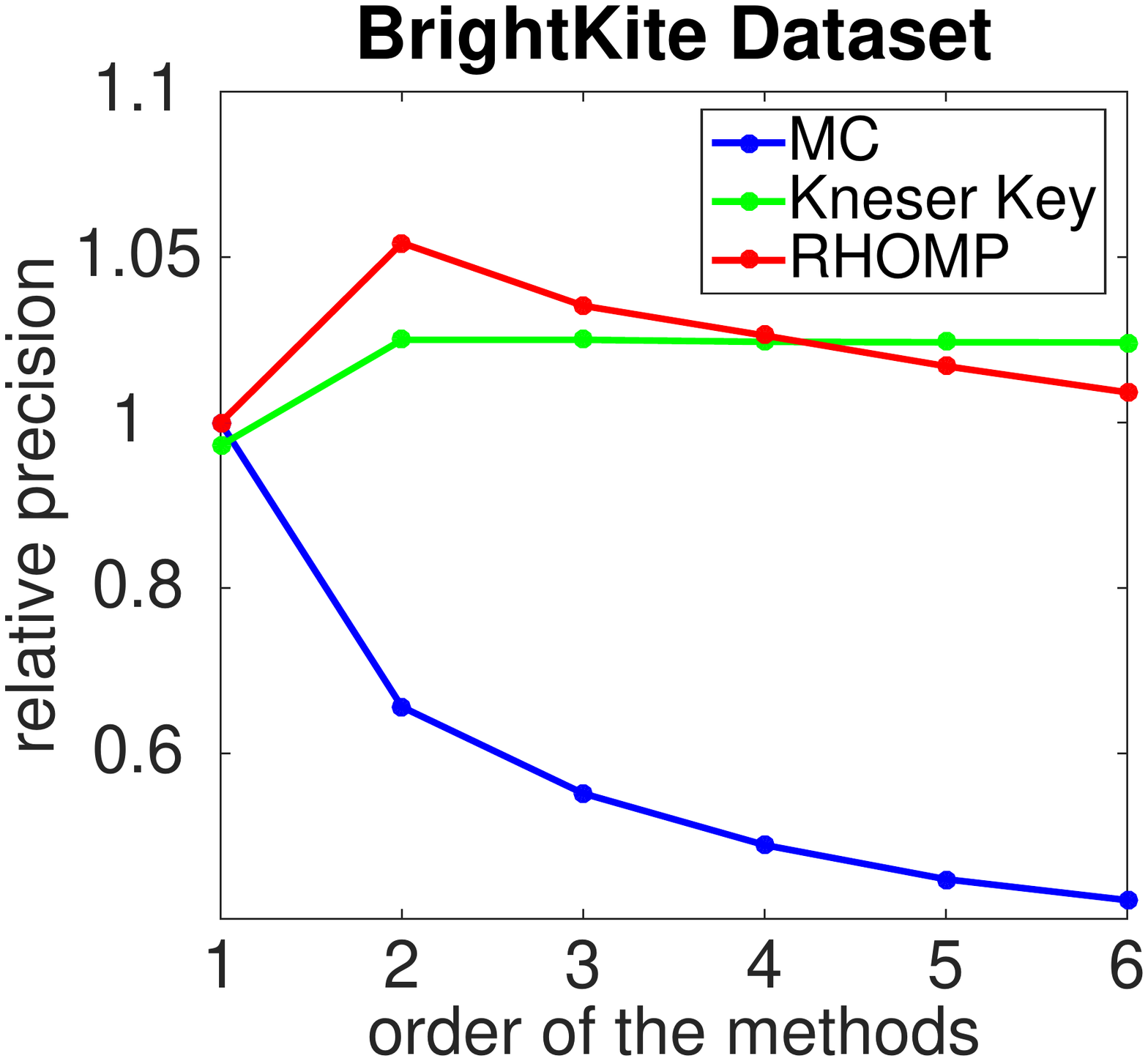}
  \includegraphics[trim={1.6cm 5.2cm 1.6cm 4.8cm}, width=.195\linewidth]{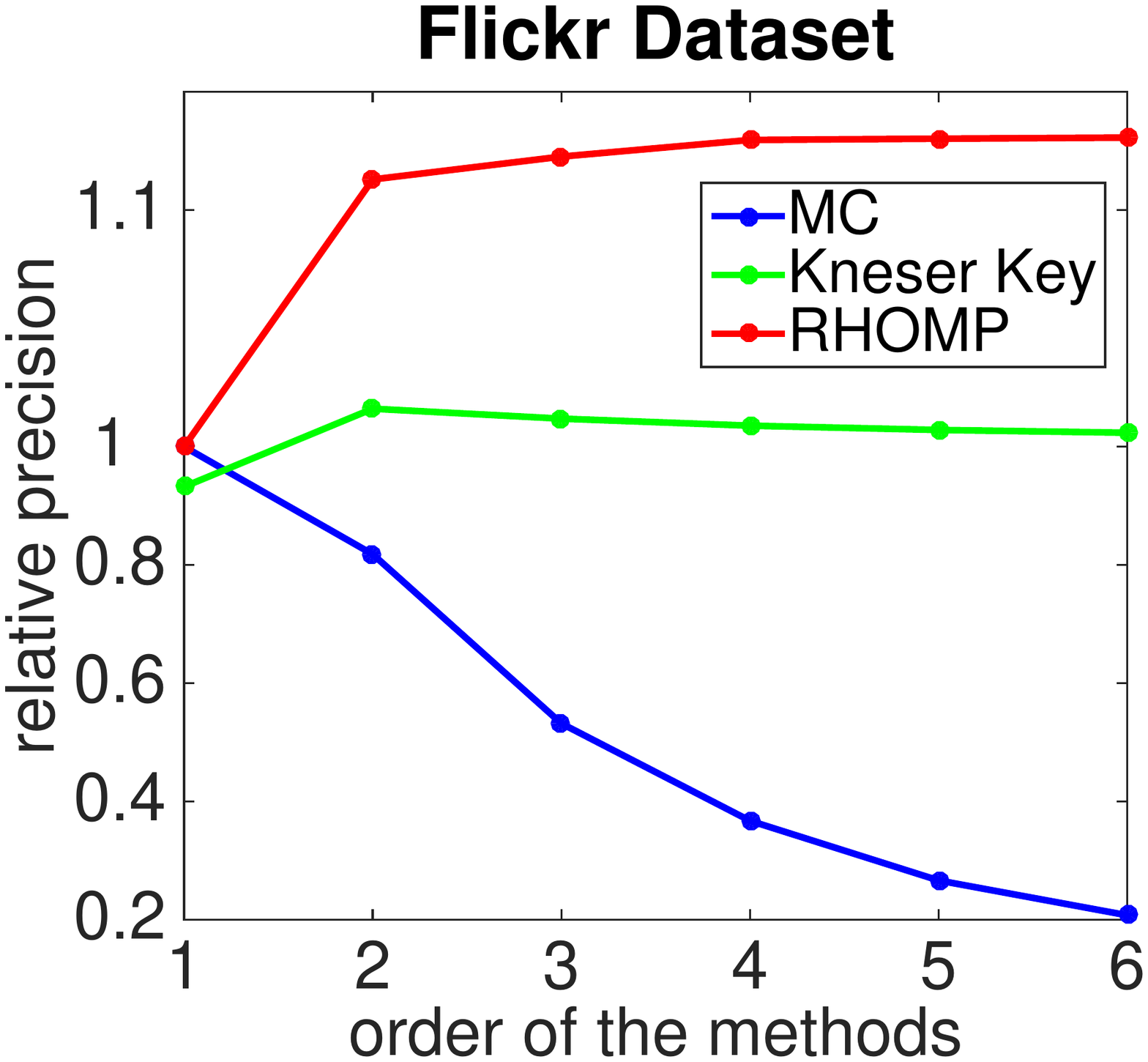}
  \includegraphics[trim={1.6cm 5.2cm 1.6cm 4.8cm}, width=.195\linewidth]{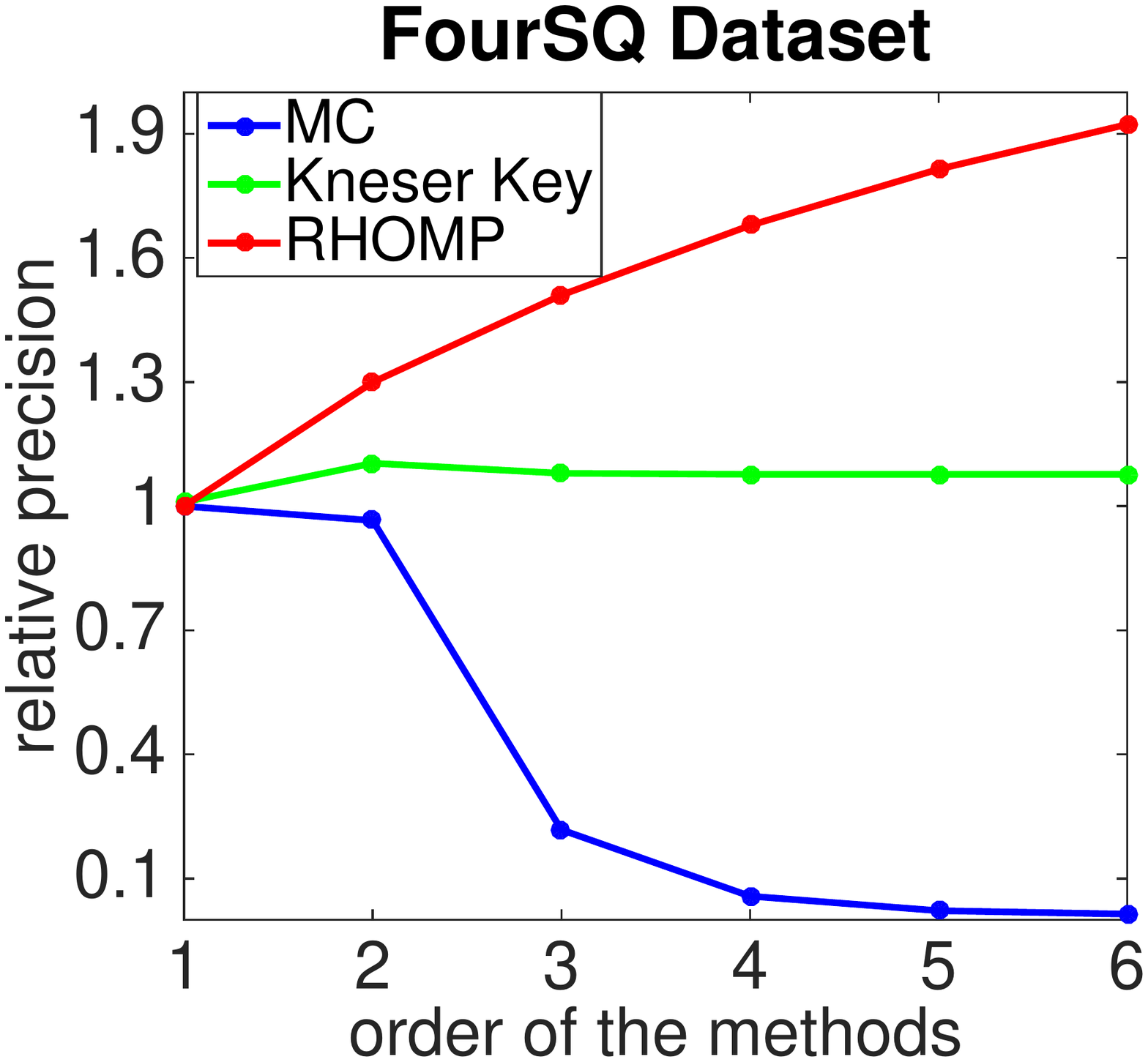}
  \caption{Relative precision ($k=3$) vs order of the methods: MC, Kneser-Ney smoothing and RHOMP.
  The relative precision is the precision ratio to that from MC1 of the corresponding datasets.
  Note that the y-axis may not be scaled linearly to make the figures more clear.}
  \vspace*{-2ex}
  \label{fig_order}
\end{figure*}

\subsection{Analysis on Higher-order Approaches} 
\label{sec_exp2}
In the previous sections, we analyze the results for first and second-order approaches. Now we study the behavior
as the order varies. Figure~\ref{fig_order} shows change in performance as the order 
increases for the three frameworks: MC, Kneser-Ney smoothing and RHOMP. For
the cases when the history states length is smaller than the order, we use the approach with the correct
order to generate the prediction.

For the MC framework, higher-order approaches make the prediction less accurate. This occurs because
these methods overfit the training data and there are more ways to overfit for a higher-order chain. 
For the Kneser-Ney smoothing approaches, in most cases (except BeerAdvocate dataset) there are
improvements moving from first-order to second-order. However the improvements are slight.
For order $>2$, there are usually either no clear improvements or small performance dips. The reason is that
as the order increase, the higher-order transition become very sparse, and could easily encounter an unseen 
higher-order state. So in this case the algorithm will frequently
seek the prediction from a lower-order approach.

For the RHOMP framework, there are improvements for each dataset when moving from MC1 to RHOMP with order $=2$,
and for order $>3$, the results further improve.
Compared to MC and Kneser-Ney smoothing frameworks, The RHOMP is more robust in 
terms of not decreasing the precision as order increases, with the exception of BrightKite dataset. 
In BrightKite, the average trail length is around $3$, so there is insufficient information to train
higher-order models and we lack the lower-order fallback in Kneser-Ney. 


\section{Summary \textit{\&} Future Work } 
\label{sec_con}
In this paper we study the problem of modeling user trails, which encode useful information
for downstream applications of user experiences, recommendations and advertising. We propose a new class of structured higher-order Markov chains which we call 
the retrospective higher-order Markov process (RHOMP). This model preserves the higher-order nature of user trails
without risks of overfitting the data. A RHOMP can be estimated from data via a projected gradient descent algorithm we propose for maximum likelihood estimation (MLE).
In the experiments, we find that RHOMP is superior in terms of precision and mean reciprocal rank compared to other methods. Also RHOMP is robust for higher-order chains when there is data available.

There are several directions to extend this work. First it would be interesting to 
explore other forms of retrospection that allow more interaction between the history states. (Note that the current
approach in this paper selects a single state during the retrospective process). This will allow to model
the case when certain combined history states have strong evidence in terms of transition patterns.
Second it would also be useful to extend this framework in terms of personalization. This can be achieved by a
tensor factorization approach or a collaborative filtering method. Lastly we also would like to 
embed time information into our prediction either by modeling the event time directly or using it as a side information to help generate a non-stationary process where the random walk behavior could change overtime.


\fontsize{8}{9}\selectfont

\textbf{Acknowledgements.} This work was supported by NSF IIS-1422918, CAREER award CCF-
1149756, Center for Science of Information STC, CCF-093937; DOE award DE-SC0014543; and the
DARPA SIMPLEX program.
\bibliographystyle{abbrv}
\bibliography{refs_simplified} 

\begin{thebibliography}{10}

\bibitem{awad2012prediction}
M.~A. Awad and I.~Khalil.
\newblock Prediction of user's web-browsing behavior: Application of {Markov}
  model.
\newblock {\em IEEE T. Syst. Man Cy. B}, 42(4):1131--1142, 2012.

\bibitem{Benson-2016-motif-spectral}
A.~Benson, D.~F. Gleich, and J.~Leskovec.
\newblock Higher-order organization of complex networks.
\newblock {\em Science}, 353(6295):163--166, 2016.

\bibitem{benson2015tensor}
A.~R. Benson, D.~F. Gleich, and J.~Leskovec.
\newblock Tensor spectral clustering for partitioning higher-order network
  structures.
\newblock In {\em SDM}, pages 118--126, 2015.

\bibitem{benson2016spacey}
A.~R. Benson, D.~F. Gleich, and L.-H. Lim.
\newblock The spacey random walk: A stochastic process for higher-order data.
\newblock {\em SIAM Rev.}, page To appear, 2017.

\bibitem{borges2007evaluating}
J.~Borges and M.~Levene.
\newblock Evaluating variable-length {Markov} chain models for analysis of user
  web navigation sessions.
\newblock {\em IEEE T. Knowl. Data En.}, 19(4), 2007.

\bibitem{buhlmann1999variable}
P.~B{\"u}hlmann, A.~J. Wyner, et~al.
\newblock Variable length {Markov} chains.
\newblock {\em The Annals of Statistics}, 27(2):480--513, 1999.

\bibitem{cadez2003model}
I.~Cadez, D.~Heckerman, C.~Meek, P.~Smyth, and S.~White.
\newblock Model-based clustering and visualization of navigation patterns on a
  web site.
\newblock {\em Data Mining and Knowledge Discovery}, 7(4):399--424, 2003.

\bibitem{celma2009music}
{\`O}.~Celma~Herrada.
\newblock Music recommendation and discovery in the long tail.
\newblock 2009.

\bibitem{chen2012playlist}
S.~Chen, J.~L. Moore, D.~Turnbull, and T.~Joachims.
\newblock Playlist prediction via metric embedding.
\newblock In {\em KDD}, pages 714--722, 2012.

\bibitem{chen2013multi}
S.~Chen, J.~Xu, and T.~Joachims.
\newblock Multi-space probabilistic sequence modeling.
\newblock In {\em KDD}, pages 865--873, 2013.

\bibitem{chen1996empirical}
S.~F. Chen and J.~Goodman.
\newblock An empirical study of smoothing techniques for language modeling.
\newblock In {\em ACL}, pages 310--318, 1996.

\bibitem{chierichetti2012web}
F.~Chierichetti, R.~Kumar, P.~Raghavan, and T.~Sarlos.
\newblock Are web users really {Markovian}?
\newblock In {\em WWW}, pages 609--618, 2012.

\bibitem{cho2011friendship}
E.~Cho, S.~A. Myers, and J.~Leskovec.
\newblock Friendship and mobility: user movement in location-based social
  networks.
\newblock In {\em KDD}, pages 1082--1090, 2011.

\bibitem{deshpande2004selective}
M.~Deshpande and G.~Karypis.
\newblock Selective {Markov} models for predicting web page accesses.
\newblock {\em ACM T. Internet Techno.}, 4(2):163--184, 2004.

\bibitem{duchi2008efficient}
J.~Duchi, S.~Shalev-Shwartz, Y.~Singer, and T.~Chandra.
\newblock Efficient projections onto the l 1-ball for learning in high
  dimensions.
\newblock In {\em ICML}, pages 272--279, 2008.

\bibitem{eddy1996hidden}
S.~R. Eddy.
\newblock Hidden {Markov} models.
\newblock {\em Current opinion in structural biology}, 6(3):361--365, 1996.

\bibitem{figueiredo2016tribeflow}
F.~Figueiredo, B.~Ribeiro, J.~M. Almeida, and C.~Faloutsos.
\newblock {TribeFlow}: Mining \& predicting user trajectories.
\newblock In {\em WWW}, pages 695--706, 2016.

\bibitem{gupta2016mixtures}
R.~Gupta, R.~Kumar, and S.~Vassilvitskii.
\newblock On mixtures of {Markov} chains.
\newblock In {\em NIPS}, pages 3441--3449, 2016.

\bibitem{kolda2009tensor}
T.~G. Kolda and B.~W. Bader.
\newblock Tensor decompositions and applications.
\newblock {\em SIAM Rev.}, 51(3):455--500, 2009.

\bibitem{kumar2017linear}
R.~Kumar, M.~Raghu, T.~Sarl{\'o}s, and A.~Tomkins.
\newblock Linear additive markov processes.
\newblock In {\em WWW}, pages 411--419, 2017.

\bibitem{markov1971extension}
A.~Markov.
\newblock Extension of the limit theorems of probability theory to a sum of
  variables connected in a chain.
\newblock 1971.

\bibitem{mcauley2013amateurs}
J.~J. McAuley and J.~Leskovec.
\newblock From amateurs to connoisseurs: modeling the evolution of user
  expertise through online reviews.
\newblock In {\em WWW}, pages 897--908, 2013.

\bibitem{melnyk2006memory}
S.~Melnyk, O.~Usatenko, and V.~Yampol'skii.
\newblock Memory functions of the additive markov chains: applications to
  complex dynamic systems.
\newblock {\em Physica A: Statistical Mechanics and its Applications},
  361(2):405--415, 2006.

\bibitem{peres2005two}
Y.~Peres and P.~Shields.
\newblock Two new {Markov} order estimators.
\newblock {\em arXiv preprint math/0506080}, 2005.

\bibitem{pirolli1999distributions}
P.~L. Pirolli and J.~E. Pitkow.
\newblock Distributions of surfers' paths through the world wide web: Empirical
  characterizations.
\newblock {\em World Wide Web}, 2(1-2):29--45, 1999.

\bibitem{rendle2010factorizing}
S.~Rendle, C.~Freudenthaler, and L.~Schmidt-Thieme.
\newblock Factorizing personalized {Markov} chains for next-basket
  recommendation.
\newblock In {\em WWW}, pages 811--820, 2010.

\bibitem{rendle2010pairwise}
S.~Rendle and L.~Schmidt-Thieme.
\newblock Pairwise interaction tensor factorization for personalized tag
  recommendation.
\newblock In {\em WSDM}, pages 81--90, 2010.

\bibitem{ron1994learning}
D.~Ron, Y.~Singer, and N.~Tishby.
\newblock Learning probabilistic automata with variable memory length.
\newblock In {\em Proceedings of the seventh annual conference on Computational
  learning theory}, pages 35--46. ACM, 1994.

\bibitem{rosvall2014memory}
M.~Rosvall, A.~V. Esquivel, A.~Lancichinetti, J.~D. West, and R.~Lambiotte.
\newblock Memory in network flows and its effects on spreading dynamics and
  community detection.
\newblock {\em Nature communications}, 5, 2014.

\bibitem{salakhutdinov2008bayesian}
R.~Salakhutdinov and A.~Mnih.
\newblock Bayesian probabilistic matrix factorization using {Markov chain Monte
  Carlo}.
\newblock In {\em ICML}, pages 880--887. ACM, 2008.

\bibitem{shi2014collaborative}
Y.~Shi, M.~Larson, and A.~Hanjalic.
\newblock Collaborative filtering beyond the user-item matrix: A survey of the
  state of the art and future challenges.
\newblock {\em ACM Computing Surveys (CSUR)}, 47(1):3, 2014.

\bibitem{singer2015hyptrails}
P.~Singer, D.~Helic, A.~Hotho, and M.~Strohmaier.
\newblock Hyptrails: A bayesian approach for comparing hypotheses about human
  trails on the web.
\newblock In {\em WWW}, pages 1003--1013, 2015.

\bibitem{taylor2014introduction}
H.~M. Taylor and S.~Karlin.
\newblock {\em An introduction to stochastic modeling}.
\newblock Academic press, 2014.

\bibitem{thomee2015new}
B.~Thomee, D.~A. Shamma, G.~Friedland, B.~Elizalde, K.~Ni, D.~Poland, D.~Borth,
  and L.-J. Li.
\newblock The new data and new challenges in multimedia research.
\newblock {\em arXiv preprint arXiv:1503.01817}, 1(8), 2015.

\bibitem{usatenko2009random}
O.~Usatenko.
\newblock {\em Random finite-valued dynamical systems: additive Markov chain
  approach}.
\newblock Cambridge Scientific Publishers, 2009.

\bibitem{wu2016general}
T.~Wu, A.~R. Benson, and D.~F. Gleich.
\newblock General tensor spectral co-clustering for higher-order data.
\newblock In {\em NIPS}, pages 2559--2567, 2016.

\bibitem{xing2010brief}
Z.~Xing, J.~Pei, and E.~Keogh.
\newblock A brief survey on sequence classification.
\newblock {\em ACM Sigkdd Explorations Newsletter}, 12(1):40--48, 2010.

\bibitem{yang2013fine}
D.~Yang, D.~Zhang, Z.~Yu, and Z.~Yu.
\newblock Fine-grained preference-aware location search leveraging crowdsourced
  digital footprints from {LBSN}s.
\newblock In {\em UbiComp}, pages 479--488, 2013.

\bibitem{zhang2015spatiotemporal}
J.-D. Zhang and C.-Y. Chow.
\newblock Spatiotemporal sequential influence modeling for location
  recommendations: A gravity-based approach.
\newblock {\em ACM Transactions on Intelligent Systems and Technology (TIST)},
  7(1):11, 2015.

\end{thebibliography}

\end{document}